\documentclass[final,onefignum,onetabnum]{siamart190516}

\usepackage{geometry} 
\usepackage{enumitem}
\usepackage{amsfonts}
\usepackage{graphicx}
\usepackage{epstopdf}
\usepackage{algorithmic}
\usepackage[algo2e]{algorithm2e} 
\ifpdf
  \DeclareGraphicsExtensions{.eps,.pdf,.png,.jpg}
\else
  \DeclareGraphicsExtensions{.eps}
\fi

\usepackage{graphicx}
\usepackage{subcaption}
\usepackage{amssymb}
\usepackage{amsmath}
\usepackage{epstopdf}
\usepackage{url}
\usepackage{mathtools}
\usepackage{dsfont,color}

\usepackage{enumitem}
\setlist[enumerate]{leftmargin=.5in}
\setlist[itemize]{leftmargin=.5in}

\newsiamremark{remark}{Remark}
\newsiamremark{hypothesis}{Hypothesis}
\crefname{hypothesis}{Hypothesis}{Hypotheses}
\newsiamthm{claim}{Claim}

 \newcommand{\h}[1]{\mathbf{#1}}

\headers{A Lifted $\ell_1 $ Framework for Sparse Recovery}{Y. Rahimi, S. H. Kang, and Y. Lou}

\title{A Lifted $\ell_1 $ Framework for Sparse Recovery\thanks{Submitted to the editors DATE.
\funding{This work was partially funded by NSF CAREER 1846690 and Simons Foundation grant 584960.}}}

\author{Yaghoub Rahimi
\thanks{School of Mathematics, Georgia Institute of Technology, Atlanta, GA 30332 
  (\email{yaghoub.rahimi@gatech.edu}).}
  \and
  Sung Ha Kang
\thanks{School of Mathematics, Georgia Institute of Technology, Atlanta, GA 30332 
  (\email{kang@math.gatech.edu}).}
  \and Yifei Lou
\thanks{Department of Mathematical Sciences, The University of Texas at Dallas, Richardson, TX 75080
  (\email{yifei.lou@utdallas.edu}).}
  }

\usepackage{amsopn}

\ifpdf
\hypersetup{
  pdftitle={A Lifted L1 Framework for Sparse Recovery},
  pdfauthor={S. H. Kang, Y. Lou, and Y. Rahimi}
}
\fi

\begin{document}


\maketitle

\begin{abstract}
Motivated by  re-weighted $\ell_1$ approaches for sparse recovery, we propose a lifted $\ell_1$ (LL1) regularization which is a generalized form of several popular regularizations  in the literature. 
By exploring such connections, we discover there are two types of lifting functions which can guarantee that the proposed approach is equivalent to  the $\ell_0$ minimization. 
Computationally, we design an efficient algorithm via the alternating direction method of multiplier (ADMM) and establish the convergence for an unconstrained formulation. Experimental results are presented to demonstrate how this generalization improves sparse recovery over the state-of-the-art.
\end{abstract}



\section{Introduction}

Compressed sensing \cite{donoho06} plays an important role in many applications including  signal processing, medical imaging,  matrix completion, feature selection, and machine learning \cite{eldar2012compressed, foucart2017mathematical,marques2018review}. One important assumption in compressed sensing is that a signal of interest is sparse or compressible.  This allows for an efficient representation of high-dimensional data only by a few meaningful subsets. Compressed sensing often involves sparse recovery from an (under-determined) linear system that can be  formulated mathematically by minimizing the $\ell_0 $ ``pseudo-norm,'' i.e., 
\begin{equation}
\arg \min_{\h x \in \mathbb{R}^n}  \|\h x\|_0 \ \ \mbox{s.t.} \ \ A \h x = \h b,
\label{model_l0}
\end{equation}
where $A  \in\mathbb{R}^{m \times n}$ is called a sensing matrix  and $\h b \in \mathbb{R}^m $ is a measurement vector.
Since the minimization problem \eqref{model_l0} is NP-hard \cite{natarajan1995sparse},  
various regularization functionals are proposed to approximate the $\ell_0 $ penalty. The $\ell_1 $ norm is widely used as a convex relaxation of $\ell_0,$ which is called LASSO \cite{tibshirani1996regression} in statistics or basis pursuit \cite{chen2001atomic} in signal processing. Due to its convex nature, the $\ell_1$ norm is computationally traceable to optimize with  exact recovery guarantees based on restricted isometry property (RIP) and/or null space property (NSP) \cite{candesRT06,donoho2001uncertainty,tillmann2014computational}. Alternatively, there are non-convex models, i.e. concave with respect to the positive cone, that  outperform the convex $\ell_1$ approach in practice. For example, $\ell_p (0<p<1)$ \cite{chartrand2007exact,Xu2012,zong2012representative}, smoothly clipped absolute deviation (SCAD) \cite{fan2001variable}, minimax concave penalty (MCP) \cite{zhang2010nearly},
capped $\ell_1$ (CL1) \cite{louYX16, shen2012likelihood,zhang2008multi},  and transformed $\ell_1 $ (TL1) \cite{lv2009unified,zhang2014minimization,zhang2018minimization}  are separable and concave penalties.    Some non-separable non-convex penalties include sorted $\ell_1 $ \cite{huang2015nonconvex}, $\ell_1-\ell_2 $ \cite{louY18,louYHX14,yin2014ratio,yinLHX14}, and $\ell_1/\ell_2 $ \cite{rahimi2019scale,wang2020accelerated}.

In this paper, we generalize a type of re-weighted approaches \cite{candes2008enhancing,guo2021novel, wipf2010iterative} for sparse recovery   based on the fact that a convex function can be smoothly approximated by subtracting a strongly convex function from the objective \cite{nesterov2005smooth}.  In particular, we propose a \textit{lifted regularization}, which we referred to as \textit{lifted $\ell_1$ (LL1)},
\begin{equation}
F^U_{g, \alpha}(\h x) = \min_{\h u\in U} \left\langle \h u, | \h x| \right\rangle + \alpha g(\h u),
\label{objective:newmodel}
\end{equation}
where we denote $|\h x | = [ |x_1|,...,|x_n|]\in\mathbb R^n$ and $\langle \cdot , \cdot \rangle$ is the standard Euclidean
inner product. In \eqref{objective:newmodel}, the variable $\h u $ plays the role of weights with $U$ as a set of restrictions on $\h u $, e.g., $U = [0,\infty)^n $ or $  U = [0,1]^n $.
The function $g:\mathbb{R}^n \rightarrow \mathbb{R}  $ is  decreasing near zero (please refer to \cref{def:01} for multi-variable decreasing function) to ensure a non-trivial solution of \eqref{objective:newmodel}, and  $\alpha > 0$ is a parameter. 
We also consider
a more general function $g(\h u; \alpha)$, instead of  $\alpha g(\h u)$ when making a connection of the proposed model to several existing sparse-promoting regularizations in   Section~\ref{sec_connection}. 
Note that ``$\min$'' is used instead of ``$\inf $'' in \eqref{objective:newmodel}, since we assume that the infimum is attained by at least one point.

To find a sparse signal from a set of linear measurements, we propose the following constrained minimization problem:
\begin{equation}\label{eq:con}
  \min_{\h x, \h u\in U}  \left\langle \h u, |\h x| \right\rangle + \alpha g(\h u)  \ \ \mbox{s.t.} \ \ A \h x = \h b.
\end{equation}
We lift the dimension of 
a single vector $\h x\in \mathbb R^n$ in the original $\ell_0$ problem \eqref{model_l0} into a joint optimization over $\h x $ and $\h u$ in the proposed model \eqref{eq:con}. 
We can establish the equivalence between the two,  if the function $g$ and the set $U$ satisfy certain conditions.  For the measurements with noise, we also consider an unconstrained formulation, i.e.,
\begin{equation}\label{eq:uncon}
  \min_{\h x, \h u\in U}  \left\langle \h u, |\h x| \right\rangle + \alpha g(\h u)  + \frac {\gamma} 2 \| A \h x - \h b\|_2^2,
\end{equation}
where $\gamma>0$ is a balancing parameter. 

From an algorithmic point of view, the lifted form enables us to  minimize two variables each with fast implementation, 
and can lead to a better local minima in higher dimension than the original formulation \cite{nesterov2005smooth,zach2017iterated}.  To minimize  \eqref{eq:con} and \eqref{eq:uncon},
we apply the alternating direction method of multipliers (ADMM) \cite{boyd2011distributed,gabay1976dual}, and conduct convergence analysis of ADMM for solving  the unconstrained problem \eqref{eq:uncon}.  We demonstrate in experiments that the proposed approach outperforms  the state-of-the-art.  Our contributions are  threefold,
\begin{itemize}
    \item We propose a unified model that generalizes many existing regularizations and re-weighted algorithms for sparse recovery problems;
    \item We establish the equivalence between the proposed model (\ref{eq:con}) and the $\ell_0$ model \eqref{model_l0};
    \item We present an efficient algorithm which is supported by a  convergence analysis.
\end{itemize}


The rest of the  paper is organized as follows. In  \Cref{sec_model}, we present details of the proposed model, including 
its properties with an exact sparse recovery guarantee.  \Cref{sec_numAlg} describes  the ADMM implementation and its convergence. 
 \Cref{sec_Experiments} presents experimental results  to demonstrate how this lifted model improves the sparse recovery over the state-of-the-art.  Finally,  concluding remarks are given in  \Cref{sec_conc}.  

\section{The Proposed Lifted \texorpdfstring{$\ell_1 $}{L1} Framework} \label{sec_model}
We first introduce definitions and notations that are used throughout the paper. The connection to well-known sparsity promoting regularizations is presented in Section~\ref{sec_connection}.  We present useful properties of LL1 in Section~\ref{sect:prop}, and establish its equivalence to the original $\ell_0$ problem in Section~\ref{sect:exact}.

\subsection{Definitions and Notations}
We mark any vector in bold, specifically $\h 1$ denotes the all-ones vector and $\h 0$ denotes the zero vector.
For a vector $\h x \in \mathbb{R}^n, $ we define $|x|_{(i)} $ as the $i $-th largest component of $|\h x| $. 
We say that a vector $\h x \in  \mathbb{R}^n $ majorizes $ \h x' \in \mathbb{R}^n$ with the notation  $|\h x| \succ |\h x'|,$ if  we have $\sum_{i \leq j} |x|_{(i)} \leq \sum_{i \leq j} |x'|_{(i)} \ \forall 1\leq j < n$ and  $\sum_{i \leq n} |x_i| = \sum_{i \leq n} |x_i'| $. For two vectors $\h x,\h y\in\mathbb R^n,$ the notation $\h x\leq \h y$ means each element of $\h x$ is smaller  than or equal to the corresponding element in $\h y;$ similarly for $\h x\geq \h y$, $\h x>\h y,$ and $\h x<\h y$. The positive cone refers to the set $\mathbb{R}^n_+ = \{ \h x \in \mathbb{R}^n \mid \h x \geq \h 0 \} $. A rectangular subset of $ \mathbb{R}^n$ is a product of intervals. We define two elementwise operators, $\max(\h x,\h y)$ and $\h x \odot \h y,$ both returning a vector form by taking maximum and multiplication respectively  for every component. 
The proposed regularization \eqref{objective:newmodel} can be rewritten using the $\delta$-functon as
\begin{equation}\label{objective:newmodel1}
    F_{g,\alpha}^U(\h x) = \min_{\h u} \left\langle \h u, | \h x| \right\rangle + \alpha g(\h u)  + \delta_{U}(\h u),
    \;\; \text{where}  \;\;
\delta_{U}(\h x) = \left\{
	\begin{array}{ll}
		0  & \mbox{if } \h x \in U \\
		+\infty & \mbox{if } \h x \not\in U.
	\end{array}\right.
\end{equation}
We use \eqref{objective:newmodel} and \eqref{objective:newmodel1} interchangeably throughout the paper.  
The constrained LL1 problem \eqref{eq:con} is equivalent to 
\begin{equation}\label{eq:con1}
    \min_{\h x, \h u} \left\langle \h u, | \h x| \right\rangle + \alpha g(\h u)  + \delta_{U}(\h u) + \delta_{\Omega} (\h x),
\end{equation}
where $\Omega = \{ \h x \in \mathbb{R}^n \mid A \h x = \h b \} $. We denote $[n]$ as the set $ \{1,2,\dots,n \} $,  $S_n $ as the symmetric group of $n $ elements, and a permutation $\pi \in S_n$ of a vector $\h x$ is defined as $\h x \circ \pi = \left(x_{\pi(1)}, \dots, x_{\pi(n)} \right).$ We summarize relevant properties of a function as follows,

\begin{definition} \label{def:01}
Let $g:\mathbb{R}^n \rightarrow \mathbb{R} \cup \{ +\infty, -\infty \} $ be a function, we say that 
\begin{itemize}
    \item the function $g$ is \textit{separable} if there exists a set of functions $\{g_i: \mathbb{R} \rightarrow \mathbb{R}\}$ for  $i \in [n]$ s.t. 
    \begin{equation*}
        g(\h x) = \sum_{i = 1}^n g_i(x_i), \quad \forall \h x =[x_1,\cdots, x_n].
    \end{equation*}
   
    \item The function $g$ is  \textit{strongly convex} with parameter $\mu > 0 $ if
    \begin{equation}\label{ineq:strongly_convex}
        g(\h y) \geq  g(\h x) + \left\langle \nabla g(\h x) , \h y - \h x \right\rangle  + \frac{\mu}{2} \| \h y - \h x\|_2^2, \quad \forall \h x, \; \h y.    \end{equation}
     \item The function $g$ is \textit{symmetric} if $ g(\h x) = g(\h x \circ \pi), \; \forall \h x \in \mathbb{R}^n $ and $\forall \pi \in S_n $.
    \item  The function $g$ is  \textit{coercive} if $ g(\h x) \rightarrow +\infty \ \ as \ \ \| \h x \| \rightarrow +\infty.$
 
    \item  The function $g$ is \textit{decreasing} on $ U$ if  $g(\h x) \leq g(\h y), \; \forall \h x \geq \h y$ and  $\h x, \; \h y \in U $.
\end{itemize}
\end{definition}


\subsection{Connections to Sparsity Promoting Models}\label{sec_connection} 
Many existing models can be understood as a special case of the proposed LL1 model \eqref{eq:con}. 
We start by two recent works. 
One is a joint minimization model \cite{zhu2020iteratively} between the weights and the vector, 
\begin{equation}
\min_{\h x, \h u}  \langle \h u, |\h x|-\epsilon\rangle \ \ \mbox{s.t.} \ \ A \h x = \h b, \ \h u \in \{0,1 \}^n,
\label{model:weighted2020}
\end{equation}
where $\epsilon > 0$ is a fixed number. With the assumption that the weights are binary, the authors \cite{zhu2020iteratively} established the equivalence between \eqref{model_l0} and \eqref{model:weighted2020} for a small enough $\epsilon$. 
Another related work is the trimmed LASSO \cite{amir2021trimmed,bertsimas2017trimmed},  
\begin{equation*} 
\min_{\|\h y\|_0 \leq k} \| \h x - \h y \|_1 = \sum_{i = k+1}^{n} |x_{(i)}| = \min_{\h u} \left\langle \h u, | \h x| \right\rangle  + \delta_{U}(\h u),
\end{equation*}
which is equivalent to the LL1 form on the right with $U = \{ \h u \in \{0,1\}^n \ \mid \ \|\h u\|_1 = n-k \} $.  In the middle, the sum is over the $n-k$ smallest components of the vector $\h x $ for a given sparsity $ k$. 

In what follows, we consider a more general form of $g(\h u; \alpha),$  as opposed to $\alpha g(\h u),$ i.e.,
\begin{equation}
F_{g}^U(\h x) = \min_{\h u} \left\langle \h u, | \h x| \right\rangle + g(\h u; \alpha)  + \delta_{U}(\h u).
\label{objective:newmodel2}
\end{equation}
As a consequence of \textit{Fenchel–Moreau's} theorem, regularizations that are concave on the positive cone are of the form \cref{objective:newmodel2}. In other words, we have the following theorem:
\begin{theorem}
\label{theorem:Lifting}
Any proper and lower semi-continuous function $J(\h x) $ that is concave on the positive cone can be lifted by a convex function $g:\mathbb{R}^n \rightarrow \mathbb{R} $ and a set $U $ such that $J(|\h x|) = F_g^U(\h x) $ as in (\ref{objective:newmodel2}). 
Here $g(\h u) :=  \sup_{\h x \geq \h 0}  \left\langle \h x, -\h u \right\rangle + J(\h x) $ and $U = \{ \h u \geq 0 \mid g(\h u) \neq +\infty \}$. 
\end{theorem}

Please refer to \nameref{sec:Supplement} for detailed computations and proof of \cref{theorem:Lifting}.
Using this idea, we relate \eqref{objective:newmodel2} to the following functionals that are widely used to promote sparsity. 




(i) The $\ell_p $ model \cite{chartrandY08} is defined as
$
J^{\ell_p} (\h x) = \sum_{i=1}^{n} |x_i| ^{p}.$
As $p \rightarrow 0 $,  $\ell_p^p $ approaches to  $\ell_0 $, while 
it reduces to   $\ell_1$   for $p=1$. Much research focuses on $p=1/2,$ due to a  closed-form solution \cite{Xu2012,zong2012representative} in the optimization process.
By choosing $ g(\h u; p) =  \frac{1-p}{p} \sum_{i = 1}^n u_i^{\frac{p}{p-1}}$ and $U = \mathbb{R}^n_+ $, we can express the $\ell_p $ regularization as
\begin{equation*}
\min_{\h u} \left\{ \sum_{i=1}^{n} \left( u_i|x_i| + \frac{1-p}{p} u_i^{\frac{p}{p-1}}  \right) \ \mid \ u_i \geq 0 \right\}.
\end{equation*}

(ii) The log-sum penalty \cite{candes2008enhancing}  is given by   $
J^{\log}_a (\h x) = \sum_{i=1}^{n} \log(|x_i| + a),
$
for a small positive parameter $a$ to make the logarithmic function well-defined.
The re-weighted $ \ell_1$ algorithm \cite{candes2008enhancing} (IRL1) minimizes  $J^{\log}_a (\h x)$, which is equivalent to \eqref{objective:newmodel2} in that
\[
 \min_{\h u} \Bigg\{ \sum_{i=1}^{n} \left( u_i|x_i| + a u_i - \log(u_i)  \right)  \mid  u_i \geq 0 \Bigg\}.
\]

(iii) Smoothly clipped absolute deviation (SCAD) \cite{fan2001variable} is defined by
\begin{equation}
J_{a, b}^{\text{SCAD}} (\h x) = \sum_{i=1}^{n} f_{a, b}^{\text{SCAD}} (x_i),
\label{fun:SCAD}
\text{ where }
f_{a, b} ^{\text{SCAD}}(t) = 
\left\{
	\begin{array}{ll}
		a |t| & \mbox{if } |t| \leq  a,
		 \\
   \frac{2 a b |t| - t^2 - a^2}{2(b - 1)} & \mbox{if } a < |t| \leq  a b,
		\\
	\frac{(b + 1)a^2}{2} & \mbox{if }  |t| > ab.
				
	\end{array}
\right.
\end{equation}
This penalty is designed to reduce the bias of the $\ell_1 $ model.
For $g(\h u, a, b) = -ab \| \h u \|_1 + (b - 1) \frac{\| \h u \|_2^2}{2} $ and $U = [0, a]^n $, we have $J_{a, b}^{\text{SCAD}}$ is equivalent to
\begin{equation*}
 \min_{\h u} \left\{ \sum_{i=1}^{n} \left( u_i|x_i| - ab u_i + (b - 1)\frac{u_i^2 }{2}  \right) \ \mid \ u_i \in [0, a] \right\}.
\end{equation*}

(iv) Mini-max concave penalty (MCP) \cite{zhang2010nearly} is defined by 
\begin{equation}
J_{a, b}^{\text{MCP}} (\h x) = \sum_{i=1}^{n} f_{a, b}^{\text{MCP}} (x_i),
\label{fun:MCP}
\text{ where  }
f_{a, b}^{\text{MCP}} (t) =
\left\{
	\begin{array}{ll}
		a |t| - \frac{t^2}{2 b}  & \mbox{if } |t| \leq ab, \\
		\frac{1}{2} b a^2 & \mbox{if }  |t| > ab.
	\end{array}
\right.
\end{equation}
Same purpose of reducing bias as SCAD, MCP consists of two piece-wise defined functions, which is simpler than SCAD.
For $g(\h u, a, b) = -b (a \| \h u\|_1 - \frac{\| \h u\|_2^2}{2}) $ and $U = \mathbb{R}^n_+ $, we can rewrite $J_{a, b}^{\text{MCP}}$ as
\begin{equation*}
  \min_{\h u} \left\{ \sum_{i=1}^{n} \left( u_i|x_i| + b (\frac{u_i^2}{2} -a u_i ) \right)  \ | \  u_i \geq 0  \right\}.
\end{equation*}

(v)  Capped $\ell_1 $ (CL1) \cite{zhang2008multi} is defined as
$
J^{\text{CL1}}_a (\h x) = \sum_{i=1}^{n} \min \{ |x_i|, a \},$
with a positive parameter $a>0.$ As $a \rightarrow 0 $ the function $F^{\text{CL1}}_a / a $  approaches to $\ell_0.$ 
The CL1 penalty is  unbiased, and has less internal parameters than SCAD/MCP.
For $g(\h u,a) = -a\| \h u\|_1 $ and $U = [0,1]^n $, the CL1 regularization can be expressed as
$$
\min_{ \h u} \Bigg\{ \sum_{i=1}^{n}\left(  u_i|x_i| -a u_i \right)  \ | \   0 \leq u_i  \leq 1  \Bigg\}.
$$
It is  similar to the model in \cite{zhu2020iteratively}, 
except for a binary restriction on $\h u$ as in \eqref{model:weighted2020}.

(vi) Transformed $\ell_1 $ (TL1) \cite{lv2009unified} is defined as
$
J_a^{\text{TL1}} (\h x) = \sum_{i=1}^{n} \frac{(a+1)|x_i|}{a + |x_i|}.
$
It reduces to $\ell_0$ for $a=0,$ and converges to $\ell_1$ as $a\rightarrow \infty.$
The TL1 regularization is Lipschitz continuous, unbiased, and  
equivalent to
$$ \min_{\h u} \Bigg\{ \sum_{i=1}^{n} \left( u_i|x_i| +a u_i - 2 \sqrt{a u_i} \right)  \ | \  u_i \geq 0  \Bigg\}.
$$

(vii)  Error function penalty (ERF) \cite{guo2021novel} is defined by
\begin{equation}
J_{\sigma}^{\text{ERF}} (\h x) = \sum_{i=1}^{n} f_{\sigma}^{\text{ERF}} (|x_i|),
\label{fun:ERF}
\text{ where  }
f_{\sigma}^{\text{ERF}} (t) = \int_{0}^{t} e^{-\tau^2 / \sigma^2} d \tau.
\end{equation}
This model is less biased than $\ell_1,$ and gives a good approximation to  $\ell_0 $  for a small value of $\sigma $. Let $h(t) =  \int_{t}^{1} \sqrt{- \log( \tau)} d \tau $, then for $g(\h u, \sigma) = \sigma \sum_{i=1}^{n} h(u_i)  $ and $U = [0,1]^n $, the ERF function is equivalent to 
\begin{equation*}
 \min_{\h u} \left\{ \sum_{i=1}^{n} \left( u_i|x_i| + \sigma  h(u_i) \right) \ \mid \ u_i \in [0, 1] \right\}.
\end{equation*}



\begin{table}[t]
    \centering
\begin{tabular}{ |c|c|c|c| } 
 \hline
 Model & Regularization &  Function $ g$ & Set $U $ \\ \hline \hline
    $ \ell_p$ & $\sum_{i=1}^{n} |x_i|^{p} $ & $\frac{1-p}{p} \sum_{i = 1}^n u_i^{\frac{p}{p-1}} $ & $ \mathbb{R}^n_+$
 \\
 \hline
 log-sum  & $  \sum_{i=1}^{n} \log(|x_i| + a)$ & $ a \| \h u \|_1 - \sum_{i=1}^n \log(u_i)$ & $ \mathbb{R}^n_+$
 \\
 \hline
 SCAD & \eqref{fun:SCAD} & $-ab \| \h u \|_1 + (b - 1) \frac{\| \h u \|_2^2}{2} $ & $[0, a]^n $
 \\
 \hline
  MCP &  \eqref{fun:MCP} & $-b (a \| \h u\|_1 - \frac{\| \h u\|_2^2}{2}) $ & $ \mathbb{R}^n_+$ \\ 
 \hline
  CL1 &  $\sum_{i=1}^{n} \min \{ |x_i|, a \} $ & $-a\| \h u\|_1$ &$ [0,1]^n$ \\ 
 \hline 
 TL1 & $ \sum_{i=1}^{n} \frac{(a+1)|x_i|}{a + |x_i|} $ & $ a \|\h u \|_1 - 2 \sum_{i} (\sqrt{a u_i})$ &  $ \mathbb{R}^n_+$ \\ 
 \hline
 ERF & \cref{fun:ERF}  &  $\sigma \sum_i h(u_i)$ &  $ [0,1]^n$ \\ 
 \hline
\end{tabular}
\caption{Summary of various regularizations and their corresponding $g$ function and $U$ set.}
    \label{tab:g_functions}
\end{table}

We summarize these existing regularizations with their corresponding $g$ function and $U$ set in  \cref{tab:g_functions}. All these  $g$ functions are separable, and hence we can plot $g$ as a univariate function. As illustrated in   \cref{fig:gplot}, each $g$ function is decreasing on a small interval near origin, thus motivating  the conditions on the function $ g$ presented in \cref{def:type_g} as well as \cref{theorem:exact01,theorem:exact02} for exact recovery guarantees.

\begin{figure}[h]
     \centering
     \begin{subfigure}[b]{0.4\textwidth}
         \centering
         \includegraphics[width=\textwidth]{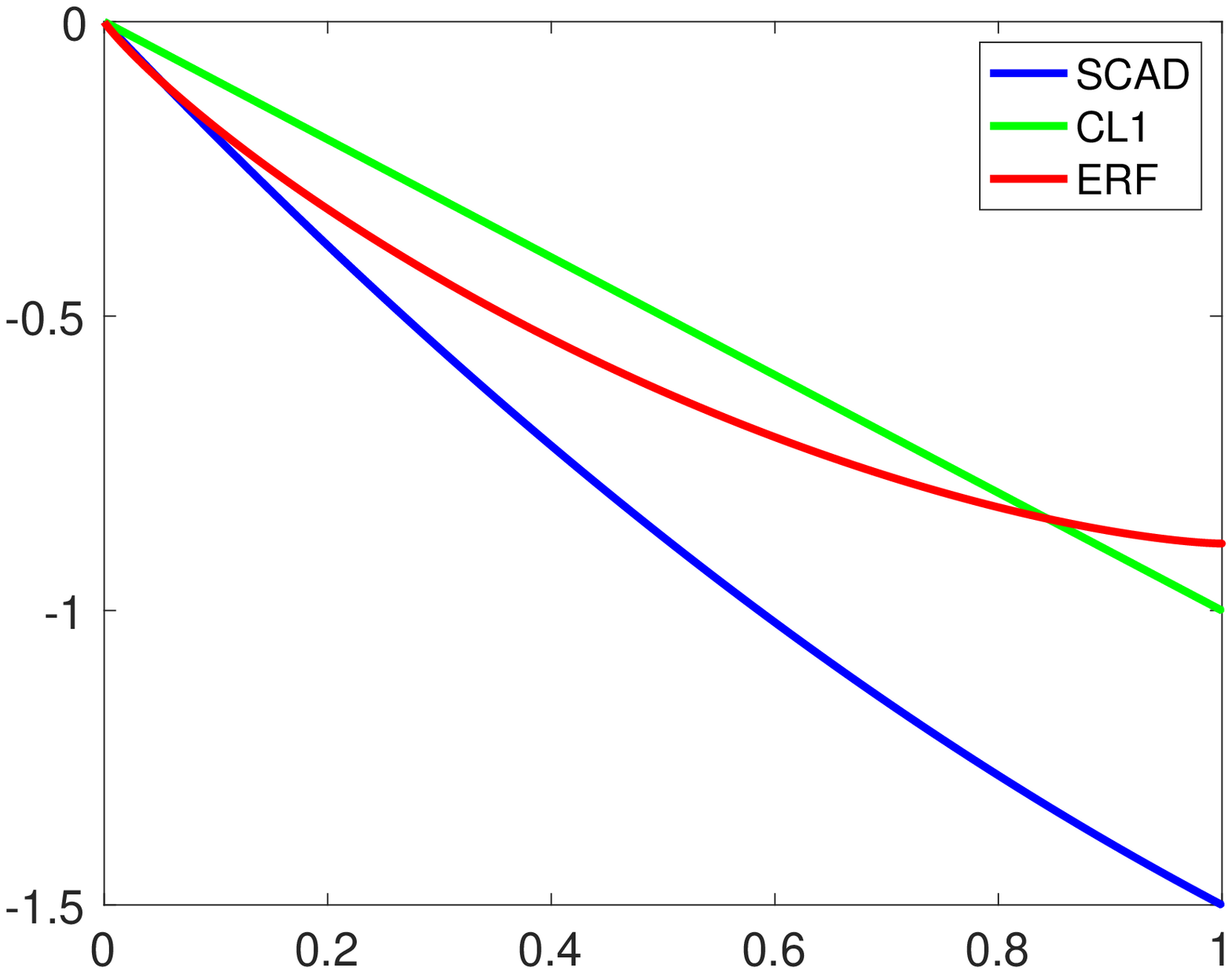}
         \caption{Decreasing functions on a bounded interval}
         \label{fig:gplot_a}
     \end{subfigure}
\hspace{0.5cm}   
     \begin{subfigure}[b]{0.4\textwidth}
         \centering
         \includegraphics[width=\textwidth]{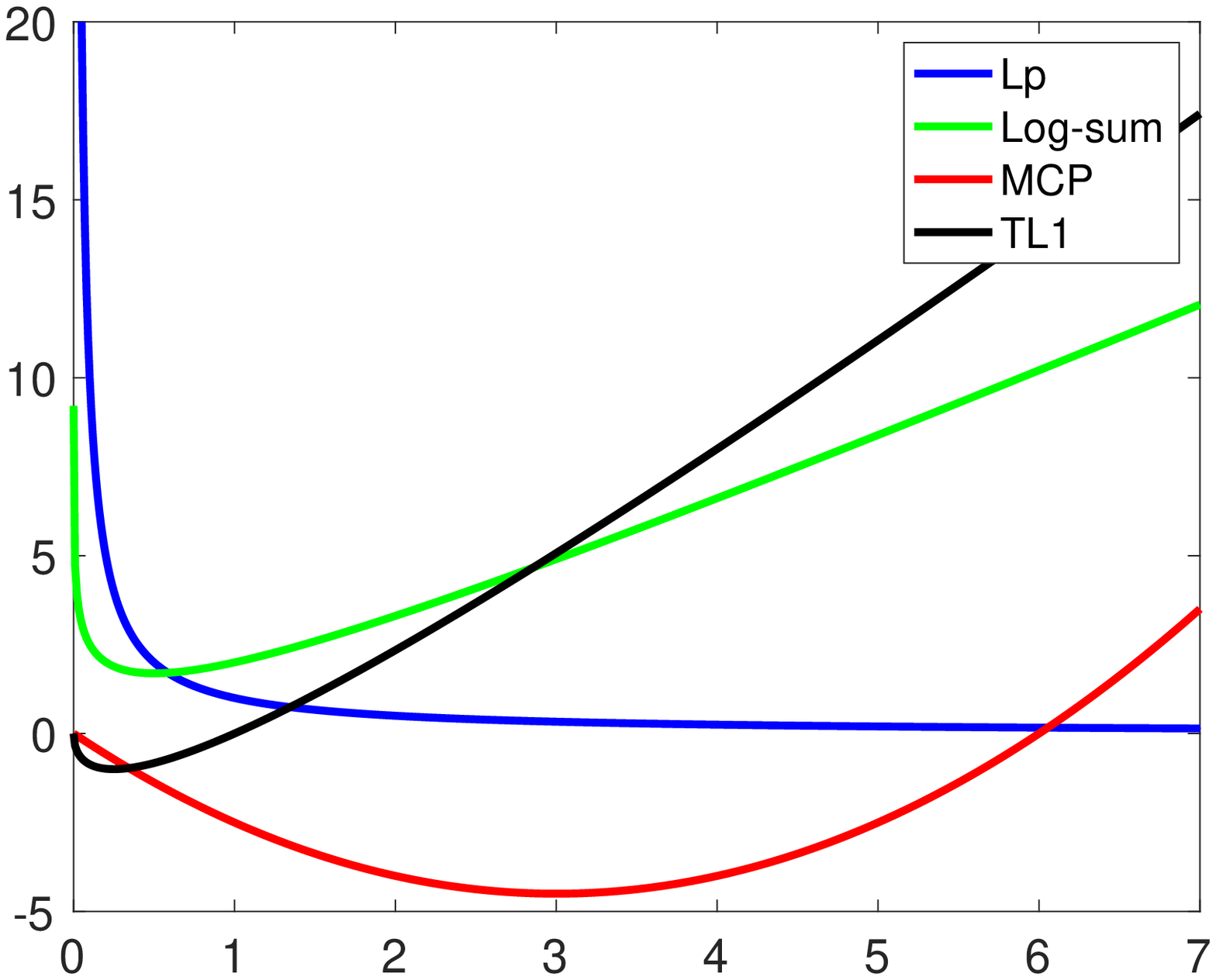}
         \caption{Convex functions on an unbounded interval}
         \label{fig:gplot_b}
     \end{subfigure}
     
        \caption{Plotting the associated $g$ function for  models listed in \cref{tab:g_functions} into two categories: (a)  SCAD with $a = 1, b = 2$, CL1 with $a = 1$, and ERF with $\sigma = 1 $;    and (b)
   $\ell_p $ with  $p = 1/2 $,  log-sum with $a = 1 $, MCP with $a = 1, b = 3 $, and  TL1 with $a = 4 $.
This motivates Type B and C in  \cref{def:type_g}.        }
\label{fig:gplot}
\end{figure}


\subsection{Properties of the 
Proposed Regularization}\label{sect:prop}

There is a wide range of analysis related to concave and symmetric regularization functions based on RIP and NSP conditions \cite{candesRT06,donoho2001uncertainty,tillmann2014computational,tran2017unified} regarding the sensing matrix $A $. Our general model $F^U_{g,\alpha}$ satisfies all the NSP-related conditions discussed in \cite{tran2017unified} so that the exact sparse recovery can be guaranteed. 
\cref{properties01} summarizes some important properties of the proposed regularization.

\begin{theorem}
\label{properties01}
For any $\h x \in \mathbb{R}^n,$ $\alpha > 0 ,$   and a feasible set $U $  on the weights $\h u, $ 
$F_{g, \alpha}^U (\h x)$ defined in \eqref{objective:newmodel} has the following properties,
\begin{enumerate}[label=(\roman*),left=\parindent]
\item{The function  $F^U_{g, \alpha}(\h x) = -\left( \alpha g + \delta_{U} \right)^*(-|\h x|),$ where $f^*$ denotes the \textit{convex conjugate} of a function $f$, thus $F^U_{g, \alpha}(\h x)$ is concave in the positive cone.
}

\item{If $g $ is symmetric on $ U $, then $ F^U_{g, \alpha}$ is symmetric on $\mathbb{R}^n $. }

\item{ If $g $ is separable on $ U $, then $ F^U_{g, \alpha}$ is  separable on $\mathbb{R}^n$.}

\item{ If $g $ is separable and symmetric on $U $, $ F^U_{g, \alpha}$  satisfies the increasing property on $\mathbb{R}^n_+ $, i.e., $F^U_{g, \alpha} (|\h x|) \geq F^U_{g, \alpha} ( |\h x'|) $ for any $|\h x| \geq |\h x'| $, and reverses the order of majorization, i.e.,   $F^U_{g, \alpha} (|\h x|) \leq F^U_{g, \alpha} ( |\h x'|) $  if $|\h x| \succ |\h x'|$.}

\item{If $g  $ is separable and $U $ is rectangular, then $ F^U_{g, \alpha}$  satisfies the sub-additive property on $\mathbb{R}^n $, i.e., 
\begin{equation*}
F^U_{g, \alpha}(\h x_1 + \h x_2) \leq F^U_{g, \alpha}(\h x_1) + F^U_{g, \alpha}(\h x_2), \quad \forall \h x_1, \h x_2 \in \mathbb{R}^n .
\end{equation*}
The equality holds  if  $\h x_1 $ and $\h x_2 $ have disjoint support, and each coordinate of $g $ has the same minimum. 

}


\item{
Let $U $ be compact and $g $ continuous. Then $F_{g,\alpha}^U $ is continuous and the set of sub-differentials of $-F_{g,\alpha}^U $ at the point $\h x \geq 0 $ is given by
\begin{equation*}
\partial  (-F_{g, \alpha}^U) (\h x)  = - \text{Conv} \left(\arg\min_{\h u \in U}   \left\langle \h u, |\h x | \right\rangle + \alpha g(\h u) \right),
\end{equation*}
where $\text{Conv} $ is the convex hall of the points. In addition, the function is differentaible at the points where the minimizer is unique. Consequently, if $g $ is strongly convex, $F_{g,\alpha}^U $ is continuously differentiable on the positive cone.
}
\item If $g(0) = 0 $ and $g$ takes its minimum value at some point in $ U$, then  we have for $\alpha_1 > \alpha_2 $ that
\begin{equation*}
    F^U_{g, \alpha_1}(\h x) \leq F^U_{g, \alpha_2}(\h x), \quad  \forall \h x \in \mathbb{R}^n .
\end{equation*}

\end{enumerate}
\end{theorem}


\begin{proof} 
\textit{(i)} Recall the convex conjugate of a function $f$ is defined as 
$
f^*(\h y)=\sup\{\langle \h x, \h y\rangle -f(\h x)\}.$
Comparing it to the definition $F^U_{g, \alpha}(\h x)$ in \eqref{objective:newmodel}, we have $F^U_{g, \alpha}(\h x) = -\left( \alpha g + \delta_{U} \right)^*(-|\h x|)$. As convex conjugate is always convex, $F_{g, \alpha}(\h x)$ is concave on the positive cone.


\textit{(ii)}  If $ g$ is symmetric, it is straightforward that  $F_{g, \alpha} $ is also symmetric.  

\textit{(iii)}
Since $g $ is a separable function,   $\min_{\h u\in U} f(\h x, \h u)$ breaks down into $n $  scalar problems, and hence  $  F^U_{g, \alpha}$ is   separable.

\textit{(iv)}
The increasing property follows from the fact that  we have for every $\h u \in U $
\begin{equation*}
 \left\langle \h u, |\h x | \right\rangle + \alpha g(\h u) \geq  \left\langle \h u, |\h x' | \right\rangle + \alpha g(\h u) \quad \forall |\h x| \geq |\h x'|.
\end{equation*}
Taking the minimum of both sides with respect to $\h u $ proves the increasing property.
The reverse order of majorization can be proved in the same way as in \cite[Proposition~2.10]{tran2017unified}.

\textit{(v)}  The sub-additive property can be proved in the same way as in   \cite[Lemma~2.7]{tran2017unified}.


\textit{(vi)}
It is straightforward that 
\begin{equation}
-F^U_{g, \alpha}(\h x) = -\min_{\h u\in U} \left\langle \h u, | \h x| \right\rangle + \alpha g(\h u)
= \max _{\h u\in U} -\left\langle \h u, | \h x| \right\rangle - \alpha g(\h u).
\end{equation}
for each $\h u \in U $. We set $f_{\h u}(\h x) = -\left\langle \h u, | \h x| \right\rangle - \alpha g(\h u)$.
As $g$ is continuous,  the map $\h u \mapsto f_{\h u}(\h x)$ is continuous for every $\h x $, and for each $\h u \in U $, the function $f_{\h u} $ is continuous. For $\h x \geq 0 $, it follows from \textit{Ioffe-Tikhomirov}'s Theorem \cite[Proposition~6.3]{hantoute2008characterizations} that 
\begin{equation}
\begin{split}
\partial  (-F_{g, \alpha}^U) (\h x)  & = \text{Conv} \left\{ \cup_{\h u \in T(\h x)} \partial f_{\h u}(\h x) \right\}
 =
\text{Conv} \left\{ \cup_{\h u \in T(\h x)} \{-\h u \} \right\}
\\ & =
\text{Conv} \{ -\h u \mid -F_{g, \alpha}^U(\h x) = f_{\h u}(\h x)  \}
 =
\text{Conv} \{ -\h u \mid \h u \in \arg\min_{\h u \in U}   \left\langle \h u, |\h x | \right\rangle + \alpha g(\h u)  \}
\\ & =
-\text{Conv} \{  \arg\min_{\h u \in U}   \left\langle \h u, |\h x | \right\rangle + \alpha g(\h u)  \},
\end{split}
\end{equation}
where $ T(\h x) = \{ \h u \mid -F_{g, \alpha}^U(\h x) = f_{\h u}(\h x)  \}$. 


 \textit{(vii)} Given $\h x \in \mathbb{R}^n $ with $\alpha_2$, we denote
$
     \h u^*_2 = \arg\min_{\h u \in U} \left\langle \h u, |\h x | \right\rangle + \alpha_2 g(\h u),
$
which implies $ F^U_{g, \alpha_2}(\h x) = \left\langle \h u^*_2, |\h x | \right\rangle + \alpha_2 g(\h u^*_2) $ and $\left\langle \h u^*_2, |\h x | \right\rangle + \alpha_2 g(\h u^*_2) \leq \left\langle \h 0, |\h x | \right\rangle + \alpha_2 g(\h 0) = 0.$ 
It further follows from $\left\langle \h u^*_2, |\h x | \right\rangle \geq 0$ that $g(\h u^*_2) \leq 0 $. For $\alpha_1 > \alpha_2 $, we have $\alpha_1 g(\h u^*_2) \leq \alpha_2 g(\h u^*_2)$ and hence
$
     F^U_{g, \alpha_1}(\h x) \leq  \left\langle \h u_2^*, |\h x | \right\rangle + \alpha_1 g(\h u^*_2) \leq \left\langle \h u_2^*, |\h x | \right\rangle + \alpha_2 g(\h u^*_2) = F^U_{g, \alpha_2}(\h x).
$
 \end{proof}
 




Using proprieties \textit{(i)-(v)}, we can prove that every $s $-sparse vector $\h x $ is the unique solution to  \eqref{eq:con} if and only if $F_{g, \alpha} $ satisfies the generalized null space property (gNSP) of order $s$. Please refer to \cite[Theorem~4.3]{tran2017unified} for more details on gNSP. As for the property \textit{(vii)}, it has algorithmic benefits, as many optimization algorithms are designed for continuously differentiable functions. 
We 
show in Theorem~\ref{thm:l0/l1} that $F^U_{g, \alpha}$ is related to $\ell_0$ and $\ell_1$  if $g$ is separable (without the assumption of strong convexity on $g$). The relationship of $F^U_{g, \alpha}$ to iteratively re-weighted algorithms, e.g., \cite{candes2008enhancing,guo2021novel} is characterized in \cref{thm:reweighted}.




\begin{theorem}\label{thm:l0/l1}
Suppose $U= [0,1]^n$ and $g$ is separable, i.e., $g(\h u) = \sum_{i=1}^n g_i(u_i)$ with each $g_i$ being a strictly decreasing  function  on $[0,1]$ with a bounded derivative. If $g_i(0) =0$ and $g_i(1) = -1 $ for $1\leq i \leq n$,  we have that for all $\h x \in \mathbb{R}^n $,
\begin{enumerate}[label=(\roman*), left=\parindent]
\item $ \frac 1 {\alpha} F^U_{g, \alpha} (\h x)  + n \rightarrow \|\h x\|_0 $, as $\alpha \rightarrow 0 $. 

\item $ F^U_{g, \alpha}(\h x) -\alpha g(\h 1)  \rightarrow \|\h x\|_1$, as $\alpha \rightarrow +\infty $.

\end{enumerate}
\end{theorem}
\begin{proof}
\textit{(i)} For any 
fixed $\h x \in \mathbb{R}^n $, we consider the derivative of $F^U_{g, \alpha}$ with respect to each of its component, i.e.,  $f_i:=|x_i| + \alpha g'(u_i)$. If $x_i = 0,$ then $f_i$ is negative due to decreasing $g,$ and hence the minimum is attained at $u_i = 1$. If $x_i \neq 0$, $f_i$ is positive 
for a small enough $\alpha,$ due to the assumption that $g'$ is bounded. Then positive derivative implies that the function is increasing, and hence the minimum is attained at $u_i = 0 $. In summary,  if $\alpha$ is sufficiently small, then we obtain that $u_i=1$ if $x_i=0$ and $u_i=0$ if $x_i \neq 0$.
For this choice of $\h u, $ we get that $ \left\langle \h u, |\h x | \right\rangle = 0$ and  $$\sum_{i} g_i(u_i) = \sum_{x_i = 0} g_i(u_i) = \sum_{x_i = 0} (-1) = \|\h x \|_0 - n. $$
which implies that $F_{g, \alpha} (\h x) =  \alpha (\|\h x \|_0 - n)$ for a small enough $\alpha $ that depends on the choice of $\h x$. By letting $\alpha\rightarrow 0$, we have $F_{g, \alpha} (\h x) / \alpha =  \|\h x \|_0 - n$ for all $\h x.$

\textit{(ii)} Since $g(1) < g(0)$ there exists a value of $u_i \in (0,1] $ with $g'(u_i) < 0 $ and so for large enough $\alpha $, the derivative  $ |x_i| + \alpha g'(u_i) $ is always negative. It further follows from  the decreasing function $g$ that the minimizer is always attained at $\h u = \h 1$ to reach to the desired result. Similarly to (i), by letting $\alpha\rightarrow +\infty,$ the analysis holds for all $\h x.$
\end{proof}

\cref{thm:l0/l1} implies that the function $ \frac{1}{\alpha}F_{g,\alpha}^U + n$ approximates the $\ell_0$ norm from below.  We can define a function of  $H(\h x, \alpha):= \frac 1 {\alpha} F^U_{g, \alpha} (\h x)  + n : \mathbb{R}^n \times [0,\alpha_0] \rightarrow \mathbb{R}$ as a transformation between $\| \h x \|_0  $ and $\frac 1 {\alpha_0} F^U_{g, \alpha_0} (\h x)  + n $ for a fixed $\alpha_0$. As characterized in \cref{cor:l0/l1}, this relationship motivates to consider a type of homotopic algorithm (discussed in \Cref{sec_numAlg}) to better approximate the desired $\ell_0$ norm, although $H(\h x, \alpha)$ is not homotopy itself (it is not continuous with respect to $\h x$).

\begin{corollary}\label{cor:l0/l1}
If $\{ \alpha_i \} $ is a decreasing sequence converging to zero and $g  $ satisfies the conditions in \cref{thm:l0/l1}, then a sequence of functions $\{ F_{g,\alpha_i}^U \}$ and  $\{ \frac{1}{\alpha_i}F_{g,\alpha_i}^U \}$ are increasing, i.e.,
\begin{equation*}
 \frac{1}{\alpha_0}F_{g,\alpha_0}^U(\h x) \leq   
 \frac{1}{\alpha_1}F_{g,\alpha_1}^U(\h x)
 \leq \dots \rightarrow \|\h x \|_0 - n, \quad \forall \h x.
\end{equation*}
\end{corollary}

If we do not restrict all the $g_i$ functions attain the same value at $1$ as in \cref{thm:l0/l1}, but rather $g_i(1)$ can take different values, then the proposed regularization
$ F^U_{g, \alpha}$ is equivalent to a weighted $\ell_1 $ model with a certain shift of $g(\h 1)$; see \cref{thm:reweighted}.

\begin{theorem}\label{thm:reweighted}
Suppose $U= [0,1]^n$ and $g$ is separable, i.e., $g(\h u) = \sum_{i=1}^n g_i(u_i)$ with each $g_i$ being a strictly decreasing  function  on $[0,1]$ with a bounded derivative. If $g_i(0) =0$ and $g_i(1) = -w_i $ for $1\leq i \leq n$,  we have  for all $\h x \in \mathbb{R}^n $,
\begin{equation}\label{eq:reweighed}
    F^U_{g, \alpha}(\h x) -\alpha g(\h 1) \rightarrow \left\langle \h w, \h x \right\rangle, \quad \text{as } \alpha\rightarrow +\infty.
\end{equation}
In another words, for a sufficiently large $\alpha$, $ F^U_{g, \alpha}$ is approaching to a weighted $\ell_1 $ model.
\end{theorem}
\begin{proof}
Since each derivative $g_i'$ is bounded, then there exists   a positive number  $M_{\h x}$ (depending on $\h x$) such that $|x_i| + \alpha g'(u_i)$ is  negative for $ \alpha > M_{\h x}$. As a result, the minimizer of  $\arg\min_{u_i}\langle u_i, |x_i|\rangle + \alpha g_i(u_i)$ is attained at $u_i = 1 $. For a sufficiently large $\alpha,$ the minimizer $\h u^* =\mathbf{1} $. By letting $\alpha\rightarrow +\infty,$ \cref{eq:reweighed} holds for all $\h x.$
\end{proof}

\subsection{Exact Recovery Analysis}\label{sect:exact}

There are many models approximating the $\ell_0$ minimization problem \eqref{model_l0}, and yet only a few of them have exact recovery guarantees. 
Motivated by the equivalence \cite{zhu2020iteratively} between the $\ell_0 $ model \eqref{model_l0} and \eqref{model:weighted2020} with a sufficiently small parameter $\epsilon$,
we  give conditions on the $g $ function to establish the equivalence between our proposed model \eqref{eq:con} and \eqref{model_l0}.  Note that  the weight vector $\h u $ in our formulation is not binary, but takes continuous values.  
By taking \cref{tab:g_functions} and \cref{fig:gplot} into account, we consider two types of $g $ functions defined as follows:

\begin{definition} \label{def:type_g}
Let $g:\mathbb{R}^n \rightarrow \mathbb{R} \cup \{ +\infty, -\infty \} $ be a separable function
with $g(\h u) = \sum_{i = 1}^n g_i(u_i),$ for $\h u =[u_1,\cdots,u_n]\in \mathbb{R}^n $.  We define 
\begin{itemize}
    \item \textbf{Type B:} All $g_i $ functions have bounded derivatives on $[0,1],$ and are strictly decreasing on $[0,1] $ with the same value at $0 $ and $ 1$, i.e. there exist two constants $a, b \in \mathbb{R} $ such that  $g_i(0) = a, g_i(1) = b,  \forall i\in[n].$
    \item \textbf{Type C:} All $g_i $ functions are convex on $[0, \infty)$ with the same value at $0 $ and with the same minimum at a point other than zero, i.e. there exist two constants $a>b \in \mathbb{R}$ such that  $g_i(0) = a $ and $\min_{t \geq 0} g_i(t) = b, \forall i\in[n]$ .
\end{itemize}
\end{definition}

An important characteristic both types of functions share is that they are decreasing near zero. 
Type B functions are defined on a bounded interval, and we enforce a box constraint on $\h u$ for strictly decreasing $g$.  
Type C refers to convex functions defined on an unbounded interval due to $\lim_{t \rightarrow \infty} g(t) > g(0).$  Note that \cref{thm:l0/l1,thm:reweighted}  hold when $g$ is a Type B or Type C function.
We establish the equivalent between \eqref{eq:con} and \eqref{model_l0} for Type B and Type C functions in \cref{theorem:exact01} and \cref{theorem:exact02}, respectively.  



\begin{theorem} \label{theorem:exact01}
Suppose $g $ is a Type B function in \cref{def:type_g} on $U = [0,1]^n $.
There exists $\alpha > 0$ such that the model \eqref{eq:con} is equivalent to \eqref{model_l0}, i.e., if $(\h x^*,\h u^*) $ is a minimizer of \eqref{eq:con}, then $\h x^* $ is a minimizer of \eqref{model_l0}; conversely, if $\h x^* $ is a minimizer of \eqref{model_l0} then by taking $u^*_i = 1 $ for $x^*_i = 0 $ and $u^*_i = 0 $ otherwise, $(\h x^*,\h u^*) $ is a minimizer of \eqref{eq:con}.
\end{theorem}
\begin{proof}
Since $g $ is a Type B function, we represent
$g(\h u) = \sum_{i=1}^n g_i(u_i)$ with each $g_i$ strictly decreasing  and having bounded derivatives on $[0,1].$
Without loss of generality, we assume that $g_i(0) = 0  $ and $g_i(1) = -1, \forall i\in [n].$
Denote
$ s := \min_{\h x} \{ \|\h x\|_0   \mid  A \h x = \h b \}$  
and $      \epsilon_0 := \min_{\h x} \{ |\h x|_{(s)} \mid A \h x = \h b \}.
$ 
Here  $\epsilon_0 > 0;$ otherwise there exists a solution to \eqref{model_l0} with sparsity less than $s$. Since $g_i $ has bounded derivatives, there exists a scalar $ \alpha > 0 $ such that $-\frac{\epsilon_0}{\alpha} < g'_i(u), \forall u \in [0,1]$ and $i \in [n] $.


Let $(\h x^*, \h u^*) $ be a solution of \eqref{eq:con}.  If $|x_i^*| \geq \epsilon_0, $ we obtain  $\frac{\partial f}{\partial u_i}=|x_i^*|+\alpha g_i'(u_i^*) > 0.$ Therefore, $h(t):= t|x^*_i| + \alpha g_i(t)$ is an increasing function on $[0,1], $ thus attaining its minimum at $t = 0 $. As a result,  we have $ u_i^* = 0$; otherwise $(\h x^*, \h u^*) $ is not a minimizer of \eqref{eq:con}. In addition, we have $ u|x| + \alpha g_i(u) \geq \alpha g_i(u) \geq \alpha g_i(1), \forall u \in [0,1]$ and $\forall x$, 
as $g_i$ is strictly decreasing. By combining two cases of $|x_i^*| \geq \epsilon_0$ and $|x_i^*| < \epsilon_0,$ we estimate a lower bound of 
\begin{equation*}
\begin{split}
f(\h x^*,\h u^*) &:= \sum_{i =1}^{n}  \Big[u^*_i|x^*_i| + \alpha g_i(u^*_i)  \Big]
 =
\sum_{\{i \mid |x^*_i| \geq \epsilon_0 \}} \alpha g_i(0)  + \sum_{\{i \mid |x^*_i| < \epsilon_0 \}}    \Big[u^*_i|x^*_i| + \alpha g_i(u^*_i)  \Big]
\\ & \geq
\sum_{\{i \mid |x^*_i| \geq \epsilon_0 \}} \alpha g_i(0)  + \sum_{\{i \mid |x^*_i| < \epsilon_0 \}}     \alpha g_i(1)  
 =
-\alpha | \{i \mid |x^*_i| < \epsilon_0 \}| \geq -\alpha (n-s),
\end{split}
\end{equation*}
where we use the assumptions of $g_i(0) = 0$  and $g_i(1) = -1$  together with $| \{ i \mid |x_i| < \epsilon_0\}| \leq n-s  $ by the definitions of $s$ and $\epsilon_0.$
On the other hand, the lower bound $-\alpha (n-s) $ for $f(\h x^*, \h u^*) $ can be achieved by any solution $\h x $ of $ A \h x = \h b$ with sparsity $ s$,  by choosing $u_i = 0 $ for $ x_i \neq 0$ and $u_i = 1 $ otherwise. Therefore, we   have $f(\h x^*, \h u^*) = -\alpha (n-s)$.

Next we show that $\h x^* $ must have sparsity $ s$. If $ |x^*_i| \in (0,\epsilon_0) $ for some $i$, then we have $u_i|x_i| + \alpha g_i(u_i)  > \alpha g_i(1).$ Note that  the inequality is strict, forcing $ f(\h x^*, \h u^*)$ to be strictly greater than the lower bound $-\alpha (n-s) $, which is a contradiction. Therefore, if $|x^*_i| < \epsilon_0,$
we must have $ x_i = 0$. Based on the definition of $\epsilon_0,$ we have $\|\h x^* \|_0 = s $, which implies $\h x^* $ is a minimizer of \eqref{model_l0}. 





Conversely, if $\h x^* $ is a solution of \eqref{model_l0}, then $\h x^*$ satisfies $A\h x^*=\h b.$ With the choice of $ u_i^* = 0$ for $ |x_i^*| \neq 0$ and  $u_i^* = 1 $ otherwise, we get $(\h x^*, \h u^*)$ is a minimizer of \eqref{eq:con} such that the objective function attains the minimum value $-\alpha (n-s) $.
\end{proof}

\begin{theorem} \label{theorem:exact02}
Suppose $U = [0,\infty)^n $ and $g $ is a Type C function in \cref{def:type_g}. Then there exists a constant $\alpha > 0$ such that the model \eqref{eq:con} is equivalent to \eqref{model_l0}. 
\end{theorem}
\begin{proof}
We define $f(\h x, \h u) $, $ s$, and $\epsilon_0 $ in the same way as in the proof of  \cref{theorem:exact01}.
Since $g_i$ is convex, $g_i' $ is increasing and hence $g_i'(u) \geq g_i'(0), \forall u \in [0,\infty)$ and $i\in [n]$. Then there exists a scalar $\alpha>0$ such that $g_i'(0) > -\frac{\epsilon_0}{\alpha}, \forall i\in [n]$.
Without loss of generality, we assume 
 $a=0,  b = \arg\min_{t \geq 0} g_i(t) =-1$. In this setting, we get
$$ u|x| + \alpha g_i(u) \geq \alpha g_i(u) \geq \alpha g_i(d_i) = -\alpha, \ \forall u \in [0,\infty), \ \forall x, \ i\in [n],$$
which implies that   $f(\h x^*, \h u^*) \geq  -\alpha(n-s).$
The rest of the proof  follows exactly from the one of  \cref{theorem:exact01}, thus omitted. 
\end{proof}

\begin{remark}
In \cref{theorem:exact01,theorem:exact02}, we consider a linear constraint set  $\Omega = \{ \h x \in \mathbb{R}^n \mid A \h x = \h b \}.$ All the analysis can be extended to a feasible set of inequality constraints, e.g., $\Omega_\epsilon = \{ \h x \in \mathbb{R}^n \mid  \|A \h x - \h b\| \leq \epsilon \} $ for $\epsilon \geq 0 $. In this case, we can show that our model \cref{eq:con1} with a given $\Omega_\epsilon$ is equivalent to the following $\ell_0$ formulation, 
\begin{equation}
\arg \min_{\h x \in \mathbb{R}^n}  \|\h x\|_0 + \delta_{\Omega_\epsilon} (\h x).
\end{equation}
\end{remark}



\section{Numerical Algorithms and Convergence Analysis}\label{sec_numAlg}

We describe  in Section~\ref{sec_admm} the alternating direction method of multipliers (ADMM) \cite{boyd2011distributed,gabay1976dual} for solving the general model,  with convergence analysis presented in Section~\ref{sec_convergence}. In Section~\ref{sec_CompG}, we discuss closed-form solutions of the $\h u$-subproblem for two specific choices of  $g$.


\subsection{The Proposed Algorithm} \label{sec_admm}

We define a function $\psi(\cdot)$ to unify the constrained  and the unconstrained formulations, i.e., $\psi(\h x) = \delta_\Omega(\h x)$ for \eqref{eq:con} and $\psi(\h x) = \frac \gamma 2 \|A\h x-\h b\|_2^2$ for \eqref{eq:uncon}. 
We introduce a new variable $\h y$ in order to apply ADMM to minimize 
\begin{equation}\label{eq:ADMMobj}
\arg \min_{\h u, \h x,\h y} \{ \left\langle \h u, |\h x | \right\rangle + \alpha g(\h u)  + \delta_{U}(\h u) + \psi(\h y)\ | \  \ \h y = \h x \}.
\end{equation}
The corresponding augmented Lagrangian becomes 
\begin{equation}\label{eq_L}
L_{\rho}(\h u, \h x, \h y; \h v) = \left\langle \h u, |\h x | \right\rangle + \alpha g(\h u) + \delta_U(\h u) + \psi(\h y)+ \left\langle \h v, \h x - \h y \right\rangle +  \frac{\rho}{2} \| \h x - \h y  \|_2^2,
\end{equation}
where $\h v$ is the Lagrangian dual variable and $\rho$ is a positive parameter. The ADMM scheme involves 
 the following   iterations,
\begin{equation}\label{alg_admm}
\left\lbrace
\begin{split}
\h u^{k+1} & = \arg\min_{\h u} L_{\rho}(\h u, \h x^{k}, \h y^{k}; \h v^k)
\\
\h x^{k+1} & = \arg\min_{\h x} L_{\rho}(\h u^{k+1}, \h x, \h y^{k}; \h v^k)
\\
\h y^{k+1} & = \arg\min_{\h y} L_{\rho}(\h u^{k+1}, \h x^{k+1}, \h y; \h v^k)
\\
\h v^{k+1} & = \h v^k + \rho(\h x^{k+1} - \h y^{k+1}).
\end{split}
\right.
\end{equation}
The original problem \eqref{eq:con} jointly minimizes  $\h u $ and $\h x ,$ which are
updated separately in \eqref{alg_admm}. 
In particular, the $\h u$-subproblem can be expressed as
\begin{equation}
\h u^{k+1} = \arg\min_{\h u \in U}   \left\langle \h u, |\h x^k | \right\rangle + \alpha g(\h u).
\label{eq:update_u}
\end{equation}
In general, one may not  find a closed-form solution to \eqref{eq:update_u}. For a separable function $g$ and a rectangular set $ U$, the $\h u $-update simplifies into $n $ one-dimensional minimization problems; refer to Section~\ref{sec_CompG} for the $\h u $-update with two specific $g$ functions that are used in experiments. For the $\h x$-update, it has a  closed-form solution given by  
\begin{equation*}
\h x^{k+1}  = \mathbf{shrink}\left( \h y^{k} - \frac{1}{\rho} \h v^k, \frac{1}{\rho} \h u^{k+1} \right),
\end{equation*}
where 
$
    \mathbf{shrink}(\h v, \h u)  = \mathrm{sign}( \h v)  \odot \max(|\h v|-\h u,\h 0).
$

For the constrained formulation, i.e., $\psi(\h y) = \delta_{\Omega}(\h y) $, the $\h y$-subproblem becomes
\[
\arg\min_{\h y} \left\lbrace \frac{1}{2} \| \h y - (\h x^{k+1} + \frac{1}{\rho} \h v^k ) \|_2^2 \mid A \h y = \h b  \right\rbrace.
\]
It is equivalent to a projection into the affine solution of $A \h x = \h b $, which has a closed-form solution,
\begin{equation*}
\h y^{k+1} = 
\left( I_n - A^T(AA^T)^{-1}A \right)\left(\h x^{k+1} +\frac{1}{\rho} \h v^k + A^{T} (AA^{T})^{-1} \h b\right).
\end{equation*}
For the unconstrained formulation, $\psi(\h y) =  \frac{\gamma}{2} \| A \h y - \h b \|_2^2 $, the $\h y$-subproblem also has a closed-form solution by solving a linear system,
\begin{equation}\label{eq:y-update}
\begin{split}
\h y^{k+1} & = \arg\min_{\h y} \frac{\gamma}{2} \| A \h y - \h b \|_2^2 +  \frac{\rho}{2} \| \h y - (\h x^{k+1} + \frac{1}{\rho} \h v^k ) \|_2^2 
\\ & =
\left( \rho I_n + \gamma A^T A \right)^{-1} \left( \rho \h x^{k+1} + \h v^k + \gamma A^T \h b \right).
\end{split}
\end{equation}

\begin{algorithm}[t]
	\SetAlgoLined
Input: $A$ and $\h b$. Set parameters: $\rho, \epsilon$ and $Max$.

	Initialize: $\h x^0 $, $\h y^0 $,  $\alpha_0 $,  \\
	\For{ $k  = 1$ to $ Max$}{
$\h u^{k+1} =  \arg\min_{\h u \in U}   \left\langle \h u, |\h x^k | \right\rangle + \alpha^k g(\h u)$
	\\
	$\h x^{k+1}  = \mathbf{shrink}\left( \h y^{k} - \frac{1}{\rho} \h v^k, \frac{1}{\rho} \h u^{k+1} \right)$
	\\
	$\h y^{k+1}  = \arg\min_{\h y} \psi(\h y) +  \frac{\rho}{2} \| \h y - (\h x^{k+1} + \frac{1}{\rho} \h v^k ) \|_2^2 $
	\\
	$\h v^{k+1}  = \h v^k + \rho (\h x^{k+1} - \h y^{k+1}) $
	\\
	$\alpha^{k+1} = (1 - \eta) \alpha^k $
	}
	\caption{Adaptive ADMM algorithm for solving the general model \eqref{eq:ADMMobj}}
	\label{Alg:01}
\end{algorithm}

The  ADMM iterations  \eqref{alg_admm} minimize the general model \eqref{eq:con1} for a fixed value of $\alpha $. Following \cref{thm:l0/l1} and  \cref{cor:l0/l1}, we consider a type of homotopy optimization (also known as continuation approach) \cite{dunlavy2005homotopy, watson2001theory} to update $\alpha$ in order to better approximate the $\ell_0$ norm.  In particular,
we gradually decrease $\alpha $  to $0,$ while optimizing \eqref{eq:ADMMobj} for each $\alpha$. 
 \cref{Alg:01} summarizes the overall iteration for the proposed approach. 


\begin{remark}
We remark that letting $\alpha$ approach to $0$ is not exactly a homotopy algorithm, as the transformation between $\| \h x \|_0  $ and $\frac 1 {\alpha_0} F^U_{g, \alpha_0} (\h x)  + n $ is not continuous.   
We observe empirically the rate that $\alpha$ decays to zero plays a critical role in the performance of sparse recovery. 
On the other hand, we should minimize 
$\frac 1 \alpha F^U_{g, \alpha_0} (\h x) +\frac \gamma 2\|A\h x-\h b\|_2^2$ to approximate the $\ell_0$ norm.  This formulation requires the inversion of $(\rho I_n+\gamma\alpha A^TA)$ for different $\alpha$ in the $\h y$-update, which is computationally expensive, as opposed to pre-computing the inverse of $(\rho I_n+\gamma A^TA)$ with a fixed value $\gamma.$
\end{remark}




\subsection{Convergence Analysis for ADMM}\label{sec_convergence}

We prove the convergence for the ADMM method \eqref{alg_admm} for the unconstrained case, i.e.,  $\psi(\h y) = \frac{\gamma}{2} \|A \h y - \h b \|_2^2 $. In addition, we assume that $ g$ is a Type B or Type C function that is continuously differentiable on $ U$.
Note that we apply an adaptive $\alpha$ update in Algorithm~\ref{Alg:01}, but the convergence analysis is restricted to a fixed $\alpha$ value.

In this case of Type C functions $ U = [0,\infty)^n$, and it follows from optimality conditions for each sub-problem in \eqref{alg_admm} that there exists 
$ \h s^{k+1} \geq 0$ and $\h p^{k+1}  \in \partial |\h x^{k+1} | $ such that
\begin{equation}\label{eq:opt}
\left\lbrace
\begin{split}
\h 0 & = |\h x^k| + \alpha \nabla g(\h u^{k+1}) - \h s^{k+1} 
\\
\h 0 &= \h u^{k+1} \odot \h p^{k+1} + \h v^k + \rho(\h x^{k+1} - \h y^{k})
\\
\h 0 & = \nabla \psi(\h y^{k+1}) - \h v^k - \rho(\h x^{k+1}-\h y^{k+1}),
\end{split}
\right.
\end{equation}
with $\h s^{k+1} \odot \h u^{k+1} = \h 0$. 

Note that for the Type B functions, $U = [0,1]^n $ and the optimality condition for $u$ sub-problem is that there exists $\h s^{k+1}, \h t^{k+1} \geq 0 $ such that $ \h s^{k+1} \odot \h u^{k+1} = 0$ and $\h t^{k+1} \odot (\h u^{k+1} - \h 1) = \h 0 $ and $ \h 0  = |\h x^k| + \alpha \nabla g(\h u^{k+1}) - \h s^{k+1} + \h t^{k+1}$.

\begin{lemma} \label{Lemma}
Suppose the sequence $\{\h u^k, \h x^k, \h y^k, \h v^k\}$ is generated by \eqref{alg_admm}, then the following inequality holds, with $C_{A} := \|A^T A\|_2$,
\[
\|\h x^{k+1} -\h y^{k+1} \|_2 \leq \frac{\gamma C_{A}}{\rho} \|\h y^{k+1} - \h y^{k} \|_2.
\]
\end{lemma}
\begin{proof}
It is straightforward that $\psi $ has Lipschitz continuous gradient with parameter $\gamma C_A.$ 
From the $\h v$-update in \eqref{alg_admm} and the last optimality condition in \eqref{eq:opt}, we have 
$
\h v^{k+1} = \nabla \psi (\h y^{k+1}),
$
and hence
$
\| \nabla \psi (\h y^{k+1}) - \nabla \psi (\h y^{k}) \|_2 \leq \gamma C_{A} \|\h y^{k+1} - \h y^{k} \|_2,
$
which implies that
\begin{equation*}
 \|\h x^{k+1} -\h y^{k+1}\|_2 = \frac{1}{\rho}  \|\h v^{k+1} -\h v^{k}\|_2 =\frac{1}{\rho} \| \nabla \psi (\h y^{k+1}) - \nabla \psi (\h y^{k}) \|_2 \leq \frac{\gamma C_A}{\rho} \|\h y^{k+1} - \h y^{k} \|_2.
\end{equation*}
\end{proof}


\begin{theorem}[Sufficient decrease condition] \label{Thm_dec2}
Suppose the sequence $\{\h u^k, \h x^k, \h y^k, \h v^k\}$ is generated by \eqref{alg_admm}. Let $\rho \geq \sqrt{2}\gamma C_A $,
then there exists a  constant $C>0$ such that the augmented Lagrangian $L_{\rho}(\h u^k, \h x^k, \h y^k; \h v^k)$ defined in \eqref{eq_L} satisfies
\begin{eqnarray}
L_{\rho}(\h u^{k+1}, \h x^{k+1}, \h y^{k+1}; \h v^{k+1}) & \leq & L_{\rho}(\h u^k, \h x^k, \h y^k; \h v^k) - \frac{\gamma}{2} \|A(\h y^{k+1} - \h y^{k}) \|_2^2 - C \|\h y^{k+1} - \h y^{k} \|_2^2
\notag\\ &&
  - \frac{\rho}{2} \| \h x^{k+1} - \h x^k \|_2^2.\label{ineq:suff-dec2}
\end{eqnarray}
\eqref{ineq:suff-dec2} implies that $L_{\rho}$ decreases sufficiently.
\end{theorem}
\begin{proof}
The $\h v$-update in \eqref{alg_admm} and  \cref{Lemma} lead to
\begin{equation*}
\begin{split}
  &L_{\rho}(\h u^{k+1}, \h x^{k+1}, \h y^{k+1}; \h v^{k+1}) - L_{\rho}(\h u^{k+1}, \h x^{k+1}, \h y^{k+1}; \h v^{k}) \\
  & =  \left\langle \h v^{k+1} - \h v^{k}, \h x^{k+1} - \h y^{k+1} \right\rangle   = \rho \|\h x^{k+1}-\h y^{k+1}\|^2_2 \leq \frac{\gamma^2 C^2_{A}}{\rho} \|\h y^{k+1} - \h y^{k} \|_2^2.
\end{split}
\end{equation*}
 Using the $\h y$-update in \eqref{alg_admm} and the last optimality condition in \eqref{eq:opt}, we have
\begin{equation*}
\begin{split}
& L_{\rho}(\h u^{k+1}, \h x^{k+1}, \h y^{k+1}; \h v^{k}) - L_{\rho}(\h u^{k+1}, \h x^{k+1}, \h y^{k}; \h v^{k}) 
\\
& = \frac{\gamma}{2} (\| A\h y^{k+1} \|_2^2 - \| A\h y^{k} \|_2^2)  + \frac{\rho}{2} (\| \h y^{k+1} \|_2^2 - \| \h y^{k} \|_2^2) + \left\langle  \gamma A^T \h b + \rho \h x^{k+1} + \h v^{k}, \h y^{k} - \h y^{k+1} \right\rangle
\\
& = \frac{\gamma}{2} (\| A\h y^{k+1} \|_2^2 - \| A\h y^{k} \|_2^2)  + \frac{\rho}{2} (\| \h y^{k+1} \|_2^2 - \| \h y^{k} \|_2^2) + \left\langle  \gamma \rho \h y^{k+1} + \gamma A^T A \h y^{k+1}, \h y^{k} - \h y^{k+1} \right\rangle
\\
& = -\frac{\gamma}{2} \| A\h y^{k+1} - A\h y^{k} \|_2^2  - \frac{\rho}{2} \| \h y^{k+1} - \h y^{k} \|_2^2.
\end{split}
\end{equation*}
As for the $\h x$-update, we get
\begin{equation*}
\begin{split}
& L_{\rho}(\h u^{k+1}, \h x^{k+1}, \h y^{k}; \h v^{k}) - L_{\rho}(\h u^{k+1}, \h x^{k}, \h y^{k}; \h v^{k}) 
\\ & = 
\left\langle \h u^{k+1}, |\h x^{k+1}| - |\h x^k| \right\rangle + \left\langle \h v^k,\h x^{k+1} - \h x^k \right\rangle + \frac{\rho}{2} \|\h x^{k+1} - \h y^k \|_2^2 - \frac{\rho}{2} \|\h x^{k} - \h y^k \|_2^2
\\ & =
\left\langle \h u^{k+1}, |\h x^{k+1}| - |\h x^k| \right\rangle - \left\langle \h u^{k+1}\odot\h p^{k+1} + \rho \h x^{k+1},\h x^{k+1} - \h x^k \right\rangle + \frac{\rho}{2} \|\h x^{k+1}  \|_2^2 - \frac{\rho}{2} \|\h x^{k}  \|_2^2
 \\ & =  
\left[ \left\langle \h u^{k+1}, |\h x^{k+1}| - |\h x^k| \right\rangle - \left\langle \h u^{k+1}\odot\h p^{k+1} ,\h x^{k+1} - \h x^k \right\rangle \right] - \frac{\rho}{2} \| \h x^{k+1} - \h x^{k}  \|_2^2 
\;\; \leq 
   - \frac{\rho}{2} \| \h x^{k+1} - \h x^{k}  \|_2^2,
 \end{split}
\end{equation*}
where the last inequality comes from the definition of subgradient. 
For the $\h u$-update, we use the fact that $\h u^{k+1} $ is a minimizer, hence 
\begin{equation*}
\begin{split}
&  L_{\rho}(\h u^{k+1}, \h x^{k}, \h y^{k}; \h v^{k}) - L_{\rho}(\h u^{k}, \h x^{k}, \h y^{k}; \h v^{k})  \leq 0
 \end{split}
\end{equation*}
Adding all these inequalities yields the desired inequality \eqref{ineq:suff-dec2} with  $C =  ( \frac{\rho}{2} -\frac{\gamma^2 C_A^2}{\rho} )$. If $\rho \geq \sqrt{2}\gamma C_A $, then $C > 0,$ leading to the sufficient decreasing of $L_\rho.$ 
\end{proof}

\begin{theorem}[Residue convergent] 
\label{theorem:res_conv2}
Suppose the sequence $\{\h u^k, \h x^k, \h y^k, \h v^k\}$ is generated by \eqref{alg_admm}. If $g$ is a Type B or Type C function with, $\rho \geq \sqrt{2}\gamma C_A$, the following hold as $ k \rightarrow \infty$
\begin{eqnarray*}
\h x^{k+1}- \h x^k \rightarrow \h 0, \;\;\; 
\h y^{k+1}- \h y^k \rightarrow  \h 0, \;\; \text{ and } \;\; \h r^{k}:= \h x^k-\h y^k \rightarrow \h 0.
\end{eqnarray*}
\end{theorem}

\begin{proof}
Since $g $ is a Type B or Type C, then it is bounded below, and hence we denote $ m_g:=\min_{\h u\in U} g(\h u).$ By telescoping summation of \eqref{ineq:suff-dec2} from $k=1$ to $N,$ we obtain  
\begin{equation*}
\begin{split}
& \sum_{k = 0}^{N} \left( \frac{\gamma}{2} \|A(\h y^{k+1} - \h y^{k}) \|_2^2 + C \|\h y^{k+1} - \h y^{k} \|_2^2 + \frac{\rho}{2} \| \h x^{k+1} - \h x^k \|_2^2 \right)
 \\ & \leq 
   \sum_{k = 0}^{N}  L_{\rho}(\h x^k, \h y^k, \h u^k, \h v^k) - L_{\rho}(\h x^{k+1}, \h y^{k+1}, \h u^{k+1}, \h v^{k+1})
    \\ & = 
  L_{\rho}(\h x^0, \h y^0, \h u^0, \h v^0) - L_{\rho}(\h x^{N+1}, \h y^{N+1}, \h u^{N+1}, \h v^{N+1})
   \;  \leq  \;
  L_{\rho}(\h x^0, \h y^0, \h u^0, \h v^0) - \alpha m_g \; < \; \infty,
\end{split}
\end{equation*}
which implies that $ \sum_{k = 0}^{\infty} \| \h x^{k+1}- \h x^k \|_2^2 \leq \infty,$ and $ \sum_{k = 0}^{\infty} \| \h y^{k+1}- \h y^k \|_2^2 \leq \infty.$
Therefore, we must have $ \| \h x^{k+1}- \h x^k \|_2^2 \rightarrow 0,$  and  $ \| \h y^{k+1}- \h y^k \|_2^2 \rightarrow 0,$ as $k \rightarrow \infty. $ It further follows from \cref{Lemma}  that $\| \h x^k-\h y^k \|_2 \rightarrow 0,$ which completes the proof.
\end{proof}

\begin{theorem} [Stationary points] \label{thm_stat2}
Suppose the sequence $\{\h u^k, \h x^k, \h y^k, \h v^k\}$ is generated by \eqref{alg_admm}. If $g\in\mathcal C^1(\mathbb R)$ is a Type B or Type C function with $\rho \geq \sqrt{2}\gamma C_A$, and $(\h u^k, \h x^k)$ is bounded, then every limit point   of $\{\h u^k, \h x^k, \h y^k, \h v^k\},$ denoted by $\{ \h u^*, \h x^*, \h y^*, \h v^* \} $  is a stationary point of $L_\rho(\h u, \h x, \h y; \h v) $ and also $\{\h x^*, \h u^*\}$ is a stationary point of \eqref{eq:uncon}.
\end{theorem}

\begin{proof}
Using  $\h v^{k} = \nabla \psi (\h y^{k})$ from \cref{Lemma}, i.e., $\h v^{k} = \gamma A^T(A\h y^{k}-\h b),$ we have 
\begin{equation*}
\begin{split}
  & \psi(\h y^{k}) + \left\langle \h v^{k}, \h x^{k} - \h y^{k} \right\rangle
 = 
    \frac{\gamma}{2} \|A \h y^{k} - \h b\|_2^2 + \left\langle \gamma A^T(A \h y^{k}-\h b), \h x^{k} - \h y^{k} \right\rangle
    \\ & =
    \frac{\gamma}{2} \|A \h x^{k} - \h b\|_2^2 - \frac{\gamma}{2} \|A (\h x^{k} - \h y^{k})\|_2^2
\; \geq \;
    \frac{\gamma}{2} \|A \h x^{k} - \h b\|_2^2 - \frac{\gamma C_A}{2}  \|\h x^{k} - \h y^{k}\|_2^2.
\end{split}
\end{equation*}
Consequently, we obtain that
\begin{equation}
\label{eq:L_lower2}
L_{\rho}(\h u^{k}, \h x^{k}, \h y^{k}; \h v^{k})
  \geq 
 \langle \h u^{k}, |\h x^{k} |  \rangle+
\alpha g(\h u^{k}) + \frac{\gamma}{2} \| A \h x^{k} - \h b\|_2^2 +
\frac {\rho-\gamma C_A} 2\|\h x^{k}-\h y^{k}\|_2^2.
\end{equation}
If $\rho \geq \gamma C_A,$ then $L_{\rho}(\h u^{k}, \h x^{k}, \h y^{k}; \h v^{k}) \geq  \alpha g(\h u^{k}) \geq \alpha m_g$, showing $L_\rho$ is lower bounded. On the other hand,  \cref{Thm_dec2} gives an upper bound of  $L_{\rho}(\h u^{k}, \h x^{k}, \h y^{k}; \h v^{k}),$ i.e., $L_{\rho}(\h u^0, \h x^0, \h y^0; \h v^0)$. 

The boundedness of $L_\rho,$  $\h u^k,$ and $\h x^k$ together with \eqref{eq:L_lower2} implies that $\h y^k$ is bounded, and hence $\h v^k$ is bounded due to $\h v^{k} = \nabla \psi (\h y^{k}).$ Then the Bolzano-Weierstrass Theorem guarantees that  there exists a subsequence, denoted by $\{ \h x^{k_s}, \h y^{k_s}, \h u^{k_s}, \h v^{k_s} \}$, that converges to a limit point, i.e. $(\h x^{k_s}, \h y^{k_s}, \h u^{k_s}, \h v^{k_s}) \rightarrow (\h x^*, \h y^*, \h u^*, \h v^*). $

By \cref{theorem:res_conv2}, we get $\h x^{k_s}-\h y^{k_s} \rightarrow \h 0, $ leading to $\h x^* = \h y^* $, and $$(\h x^{k_s-1}, \h y^{k_s-1}) \rightarrow (\h x^*, \h y^*),$$ leading to $\h v^{k_s-1} \rightarrow \h v^* $. 

Let $\h p^{k_s} $ be the corresponding variables in the optimality condition \eqref{eq:opt}.  As $\h p^{k_s}  \in \partial |\h x^{k_s} |,$ we know $\h p^{k_s} $ is bounded by $[-1,1].$ Therefore, there exists a limit point of the sequence $\h p^{k_s}.$ Without loss of generality, we assume it is the sequence itself, i.e., $\h p^{k_s}\rightarrow \h p^*,$ and hence we have  $\h p^*  \in \partial |\h x^* | $.  

\textbf{Type C:} If $ g$ is a type C function then $\h u \in [0,\infty)^n $ and hence the optimality condition for the $\h u- $update is $\h 0  = |\h x^k| + \alpha \nabla g(\h u^{k+1}) - \h s^{k+1},$ with $\h s^{k+1} \odot \h u^{k+1} = \h 0$. We define 
\begin{equation*}
    \h s^* := \lim_{s \rightarrow \infty} |\h x^{k_s-1}| + \alpha \nabla g(\h u^{k_s}),
\end{equation*}
and  so $\h s^* \geq 0 $ and  $\h s^{k_s} \rightarrow \h s^*$ (since $g $ is continuously differentiable). The optimality condition $\h s^{k_s} \odot \h u^{k_s} = \h 0$ implies that $\h s^* \odot \h u^* = \h 0$.
Since all the equations in  \eqref{eq:opt} are continuous, we can replace $k $ by $k_s-1 $ and take the limit as $k_s \rightarrow \infty $ to get
\begin{equation*}
\left\lbrace
\begin{split}
\h 0 & = |\h x^*| + \alpha \nabla g(\h u^{*}) - \h s^{*} 
\\
\h 0 &= \h u^{*} \odot \h p^{*} + \h v^* + \rho(\h x^{*} - \h y^{*})
\\
\h 0 & = \nabla \psi(\h y^{*}) - \h v^* - \rho(\h x^{*}-\h y^{*}),
\end{split}
\right.
\end{equation*}
where $\h s^* \geq 0 $ with $\h s^* \odot \h u^* = \h 0$, and $\h p^*  \in \partial |\h x^* | $. Hence, $(\h x^*, \h y^*, \h u^*, \h v^*)$ is a stationary point of $L_\rho(\h u, \h x, \h y; \h v) $.
Furthermore, we have 
 $\h v^{k_s} = \nabla \psi (\h y^{k_s}) $ from the proof of \cref{Lemma}, leading to  $\h v^* =  \nabla \psi (\h y^*)$. Together with $\h x^* = \h y^*, $ we get  
\begin{equation*}
\left\lbrace
\begin{split}
\h 0 & = |\h x^*| + \alpha \nabla g(\h u^{*}) - \h s^{*},
\\
\h 0 & = \h u^{*} \odot \h p^{*} + \nabla \psi (\h y^*),
\end{split}
\right.
\end{equation*}
which means that $(\h x^*, \h u^*) $ is a stationary point of \eqref{eq:uncon} for $U = [0,\infty)^n $.

\textbf{Type B:} If $ g$ is a type B function then $\h u \in [0,1]^n $ and hence the optimality condition for the $\h u- $update is $\h 0  = |\h x^k| + \alpha \nabla g(\h u^{k+1}) - \h s^{k+1} + \h t^{k+1},$ with $\h s^{k+1} \odot \h u^{k+1} = \h 0$ and $\h t^{k+1} \odot (\h u^{k+1}-1) = \h 0$ with $\h s^{k+1} \geq 0 $ and $\h t^{k+1} \geq 0 $.

Note that we have $\h s^{k_s} - \h t^{k_s} = |\h x^{k_s-1}| + \alpha \nabla g(\h u^{k_s})$.
Since $\h u^{k_s} \rightarrow \h u^*,$ $\h x^{k_s-1} \rightarrow \h x^*,$ and $g $ is continuously differentiable hence the sequence $\h s^{k_s} - \h t^{k_s} $ is bounded and converges to the limit $ |\h x^*| + \alpha \nabla g(\h u^*)$. Combining the boundedness of $\h s^{k_s} - \h t^{k_s} $ together with the optimality conditions, the sequences $\h s^{k_s} $ and $\h t^{k_s} $ must be bounded. Therefore, each sequence has a convergent sub-sequence and without loss of generality, we may assume it is the sequence itself, i.e., $\h s^{k_s}\rightarrow \h s^*,$ and $\h t^{k_s}\rightarrow \h t^*$. We must have
\begin{equation*}
    \h s^* - \h t^* = \lim_{s \rightarrow \infty} |\h x^{k_s-1}| + \alpha \nabla g(\h u^{k_s}),
\end{equation*}
and $\h s^* \geq 0 $ and  $\h t^* \geq 0 $ with the conditions $\h s^* \odot \h u^* = \h 0$ and $\h t^* \odot (\h u^*-1) = \h 0$. The rest of the analysis is similar to the Type C functions and we get
\begin{equation*}
\left\lbrace
\begin{split}
\h 0 & = |\h x^*| + \alpha \nabla g(\h u^{*}) - \h s^{*} + \h t^* 
\\
\h 0 &= \h u^{*} \odot \h p^{*} + \h v^* + \rho(\h x^{*} - \h y^{*})
\\
\h 0 & = \nabla \psi(\h y^{*}) - \h v^* - \rho(\h x^{*}-\h y^{*}),
\end{split}
\right.
\end{equation*}
which means $(\h x^*, \h y^*, \h u^*, \h v^*)$ is a stationary point of $L_\rho(\h u, \h x, \h y; \h v) $ and 
\begin{equation*}
\left\lbrace
\begin{split}
\h 0 & = |\h x^*| + \alpha \nabla g(\h u^{*}) - \h s^{*} + \h t^{*},
\\
\h 0 & = \h u^{*} \odot \h p^{*} + \nabla \psi (\h y^*),
\end{split}
\right.
\end{equation*}
which means that $(\h x^*, \h u^*) $ is a stationary point of \eqref{eq:uncon} for $U = [0,1]^n $.
\end{proof}

\subsection{Algorithm updates for different \texorpdfstring{$g $}{g} functions}  \label{sec_CompG}

Here we consider two examples of $g $ functions, with which the $\h u$-subproblem has a closed-form solution. We define one function as $g_1(\h u) = -\frac{1}{2} \| \h u\|_2^2 $, a Type B function, with $U_1 = [0,1]^n$, and a Type C function $g_2(\h u) = \frac{1}{2} \| \h u\|_2^2 - \|\h u \|_1 $ with $U_2 = [0,\infty)^n$.  For these combinations the update for \eqref{eq:update_u} simplifies to
\begin{equation*}
\text{For } g_1, U_1, \;\;
u^{k+1}_i =  
\left\{
\begin{array}{cc}
    1 & \mbox{if } |x_i^k| \leq  \frac{\alpha}{2}, \\
    0 & \mbox{if } |x_i^k| \geq  \frac{\alpha}{2},
\end{array}
\right.
\;\; \text{and} \;\;\text{ for }g_2, U_2, \;\; 
 u^{k+1}_i =  \max \left\{ 1 - \frac{|x_i^k|}{\alpha}, 0 \right\}.
\end{equation*}
%
Note that for this choice of $ g_2$, the proposed model simplifies to 
\begin{equation*}
  \min_{\h x \in \mathbb{R}^n, \h u\in \mathbb{R}^n_+}  \left\langle \h u, |\h x| \right\rangle + \alpha (\frac{1}{2} \|\h u\|_2^2 - \|\h u \|_1)   \ \ \mbox{s.t.} \ \ A \h x = \h b,
\end{equation*}
which can be solved by a quadratic programming with linear constraints and positive semi-definite matrix.

    


%
%

\section{Numerical Experiments} \label{sec_Experiments}

We demonstrate the performance of  \cref{Alg:01} with $\epsilon=0.01$ and two specific $g$ functions discussed in Section \ref{sec_CompG}.
We compare with the following sparsity promoting regularizations: $\ell_1 $ \cite{chen2001atomic}, $\ell_p $ \cite{chartrand2008iteratively,Xu2012},  transformed $\ell_1 $ (TL1) \cite{zhang2014minimization}, $\ell_1 - \ell_2 $ \cite{yinLHX14}, and ERF \cite{guo2021novel}. For each model, we consider  both constrained and unconstrained formulations. Specifically for the $\ell_p$ model, we adopt the iteratively reweighted least-squares (IRLS) algorithm  \cite{chartrand2008iteratively} in the constrained case, and use the half thresholding \cite{Xu2012} as a proximal operator for minimizing the unconstrained $\ell_{1/2}$ formulation. Both $\ell_1 - \ell_2$ and TL1 are minimized by the difference of convex algorithm (DCA) for the best performance as reported in 
\cite{yinLHX14,zhang2014minimization}. We use the online code provided by the authors  of \cite{guo2021novel}  to solve for the ERF model. We use the default values of model parameters suggested in respective papers; note that $\ell_1$ and $\ell_1 - \ell_2 $ do not involve any parameters. 
All the experiments are conducted on a Windows desktop with CPU (Intel
i7-6700, 3.19GHz) and MATLAB (R2021a).




\subsection{Constrained Models}

\begin{figure}[t]
\centering
\begin{subfigure}{0.3\linewidth}
\includegraphics[width=\textwidth]{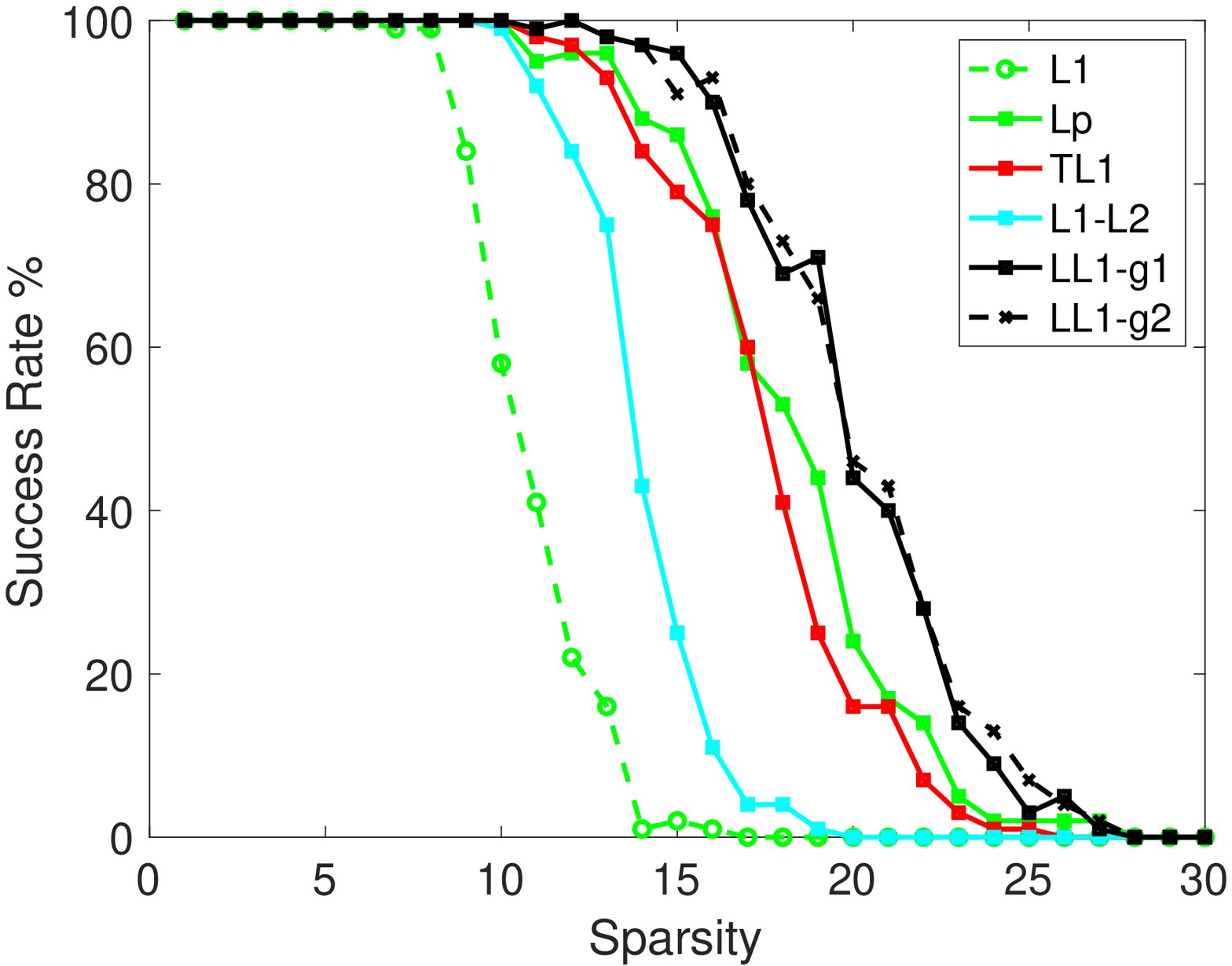}
\caption{$r = 0 $}
\label{fig:subG1}
\end{subfigure}%
~
\begin{subfigure}{0.3\linewidth}
\centering
\includegraphics[width=\textwidth]{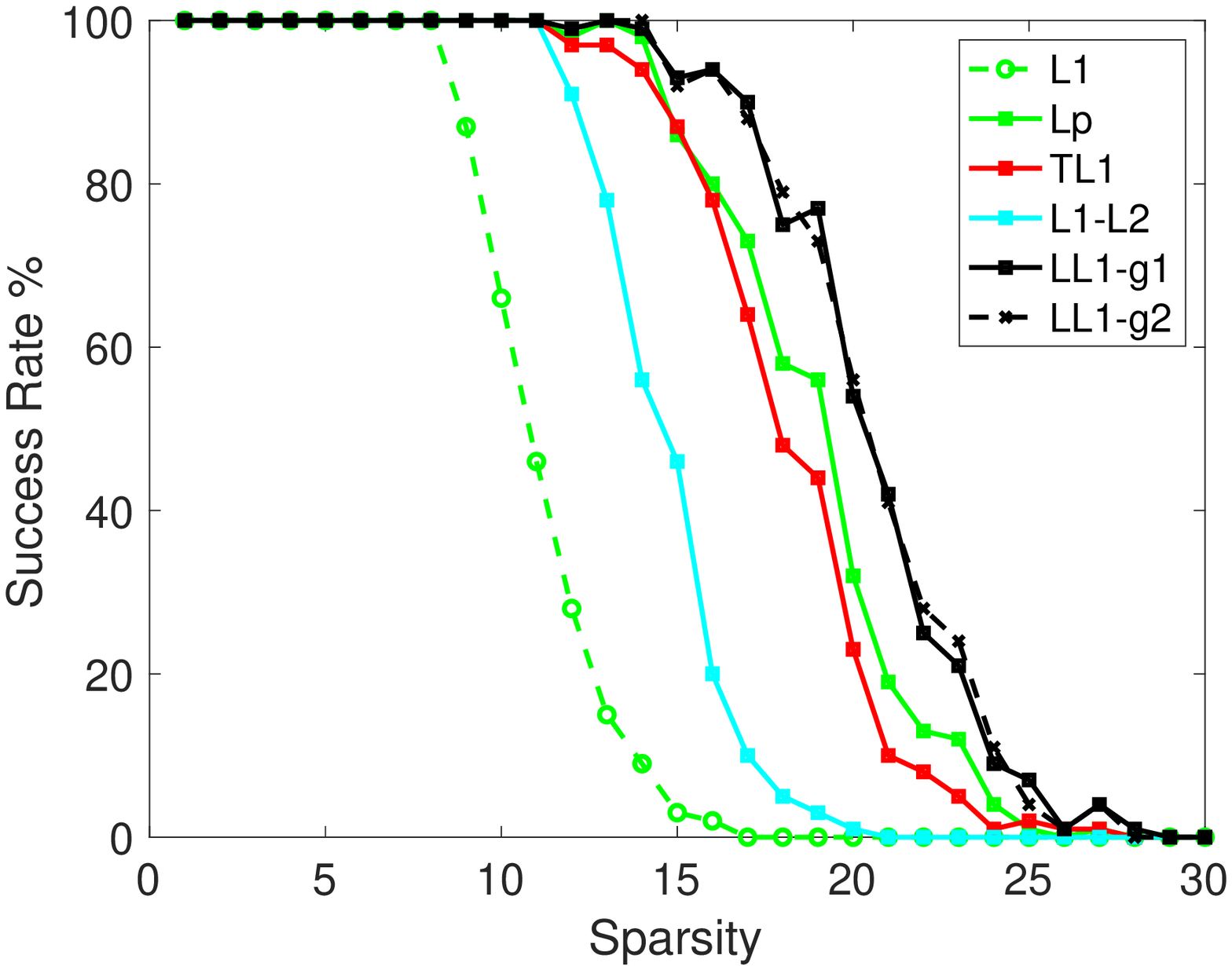}
\caption{$F = 1 $}
\label{fig:subDCT1}
\end{subfigure}
\\
\begin{subfigure}{0.3\linewidth}
\centering
\includegraphics[width=\textwidth]{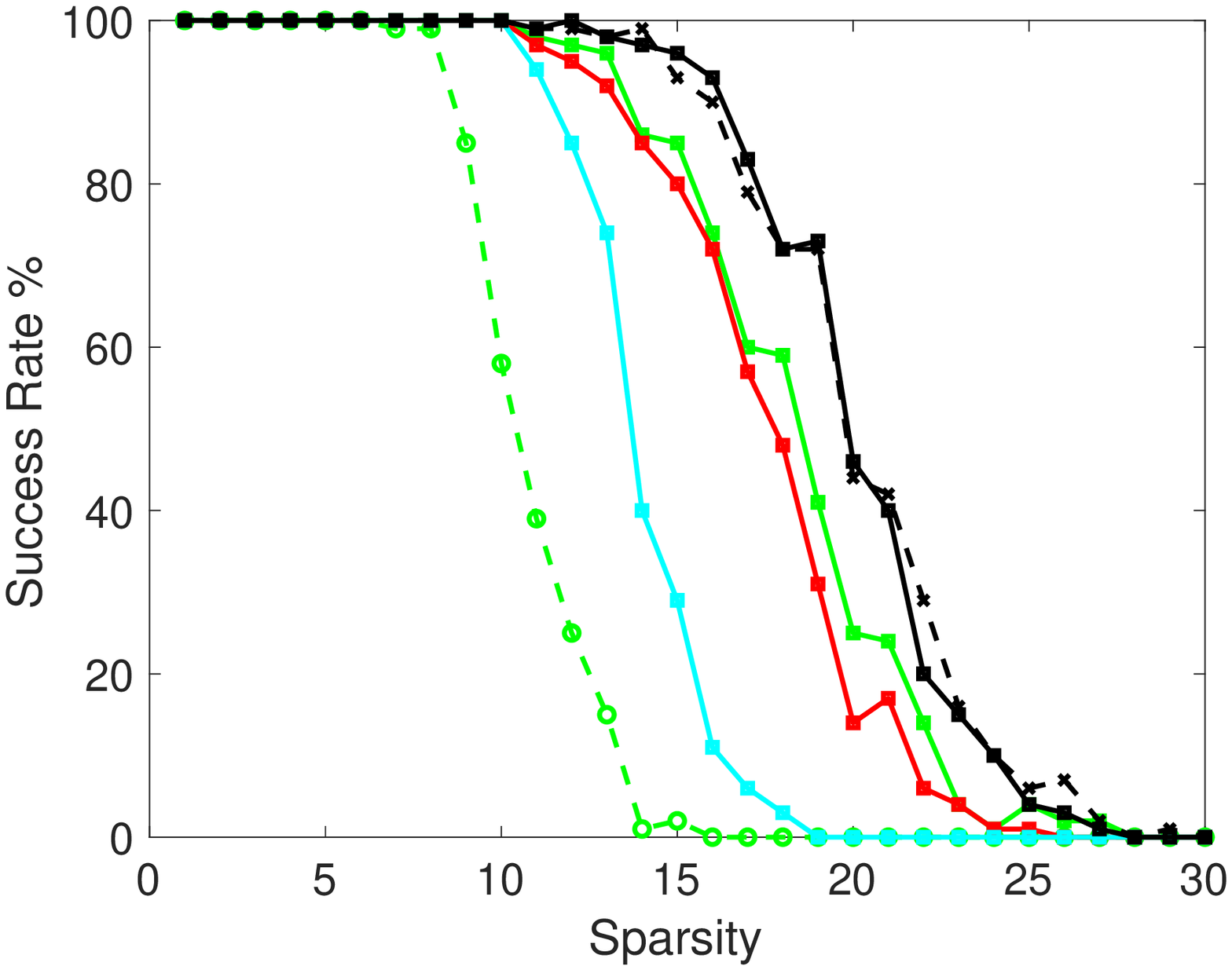}
\caption{$r = 0.8 $}
\label{fig:subG3}
\end{subfigure}
~
\begin{subfigure}{0.3\linewidth}
\centering
\includegraphics[width=\textwidth]{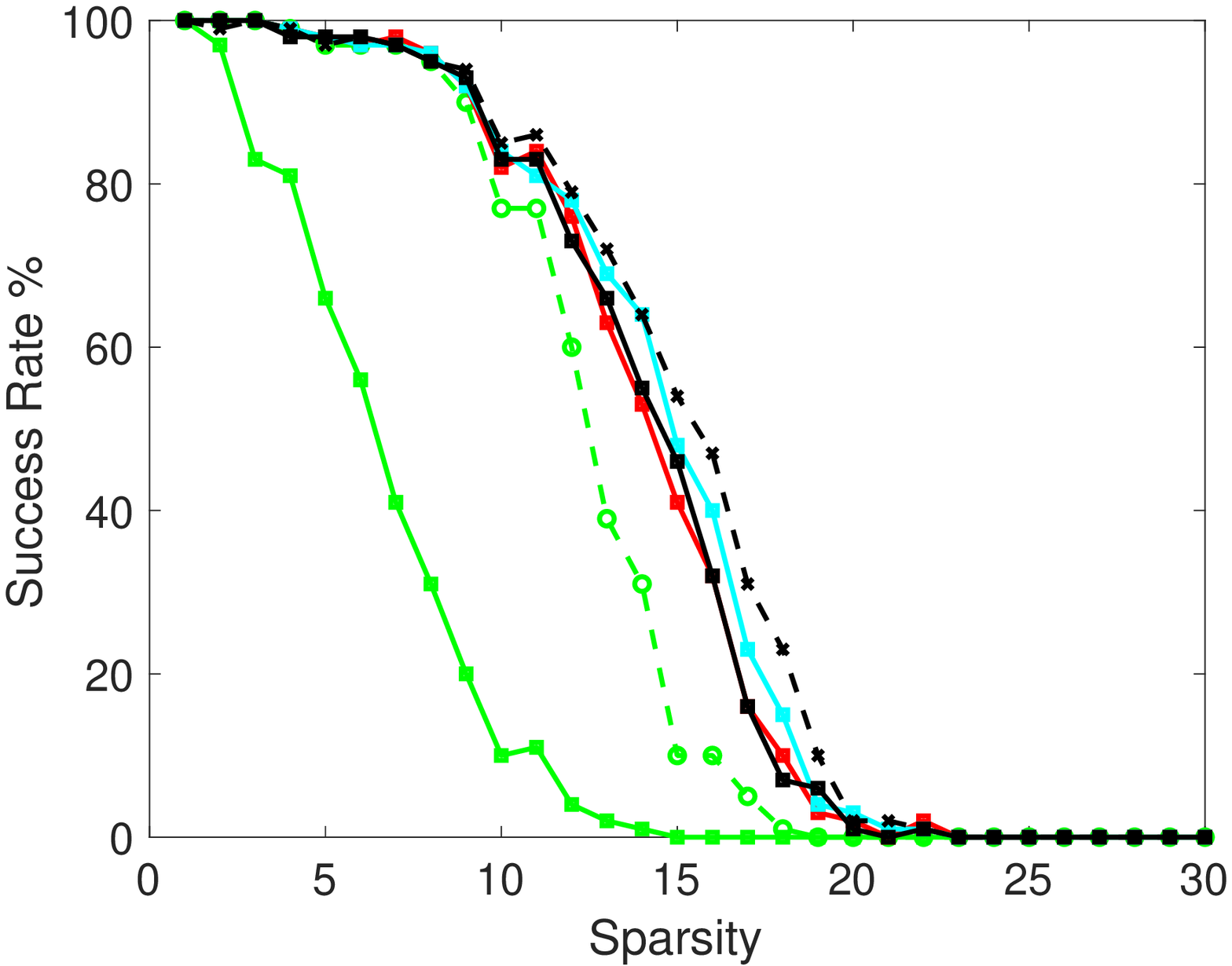}
\caption{$F = 10 $}
\label{fig:subDCT3}
\end{subfigure}
\caption{Success rate comparison among all the competing methods based on   Gaussian matrices (left) with  $r = 0, 0.8$ and DCT matrices  (right) with  $F = 1, 10$. }
\label{fig:SR_GaussianDCT}
\end{figure}

We examine the performance of finding a sparse solution that satisfies the constraint $A \h x = \h b$. 
We consider two types of sensing matrices, Gaussian and over-sampled discrete cosine transform (DCT).  The Gaussian matrix is generated based on the multivariate normal distribution $\mathcal{N}(0, \Sigma) $, where $\Sigma_{i,j} = (1-r) \delta(i = j) + r  $ for a parameter $r > 0.$
The over-sampled DCT matrix is defined by $A = [\h a_1, ... , \h a_n]  \in \mathbb{R}^{m \times n}$ with each column defined as
\begin{equation*}
\h a_j := \frac{1}{\sqrt{m}} \cos \left( \frac{2 \pi \h w_j}{F}  \right),
\end{equation*}
where $\h w $ is
a uniformly random vector and $F \in \mathbb{R}_+ $ is a scalar. The larger $F$ is, the larger the coherence of the matrix $A$ is, thus more challenging to find a sparse solution.

We fix the dimension as $64 \times 1024 $ for Gaussian and DCT matrix, while generating  Gaussian matrices with $ r \in \{0, 0.2, 0.8 \}$ and DCT matrices with $F \in \{1, 5, 10 \}$. The ground truth vector $ \h x_g \in  \mathbb{R}^n$ is simulated as $s$-sparse signal, where $s$ is the total number of nonzero entries each drawn from normal distribution $\mathcal{N}(0, 1) $ and the support index set is also drawn randomly.
We evaluate the performance by success rates where
a ``successful'' reconstruction refers to the case when the distance of the output vector $\h x $ and the ground truth $\h x_g $ is less than $10^{-2} $, i.e. $$ \frac{\|\h x - \h x_g \|_2}{\| \h x_g \|_2} \leq 10^{-2}. $$

\cref{fig:SR_GaussianDCT} presents success rates for both Gaussian and  DCT matrices, and   
demonstrates that the proposed  LL1 outperforms the state of the art in all the testing cases. For the Gaussian matrices, the parameter $r$ has little affect on the performance, as we observe the same ranking of these models under various $r$ values. As for the DCT matrices, the parameter $F$ influences the coherence of the resulting matrix. For smaller $F$ value, $\ell_p$ is the second best, while TL1 and $\ell_1-\ell_2$ perform well for coherent matrices (for $F=10$). With a well-chosen $g$ function, the proposed LL1 framework always achieves the best results among the competing methods.  
The results of LL1 using $g_1$ with $U_1$ and $g_2$ with $U_2$  are similar. 
This phenomenon illustrates  that our model works best as it is equivalent to the $\ell_0 $ model for small enough $\alpha$. 


\subsection{Unconstrained Models}

We consider the unconstrained $\ell_0$ model for comparison on noisy data:
\begin{equation}
\arg \min_{\h x \in \mathbb{R}^n}  \|\h x\|_0 + \frac{\gamma}{2} \| A \h x - \h b \|_2^2,
\end{equation}
where $\gamma $ is a regularization parameter.
We consider signals of length $512  $ with sparsity $130 $, and 
$m $ measurements $\h b $,  determined by a Gaussian sensing matrix $A $. The columns of $A $ are normalized with mean zero and unit norm. A Gaussian noise with means zero and standard deviation $\sigma$ is also added to the measurements. To evaluate the success rate of algorithms, we consider the mean-square-error (MSE) of the output signal  $\h x$ with the ground-truth solution $\h x^* $ using the formula
\begin{equation*}
\text{MSE}(\h x) = \|\h x - \h x^* \|_2.    
\end{equation*}
For each algorithm, 
we compute the average of  MSE for 100 realizations by ranging the number of measurements between $60 <m <120 $. 
\cref{fig:Noisy_comALL01} present the comparison results for two noise levels $\sigma \in \{10^{-6}, 0.01 \} $.
All the algorithms perform badly with a few measurements, and as the number of measurements $m $ increases, their MSE error decreases. For the smaller amount of the noise ($\sigma=10^{-6}$), our approach almost works perfectly in  around $100$ measurements, while other algorithms either require more measurements to achieve the nearly perfect MSE or are unable to do so.
\cref{fig:Noisy_comALL01_d} presents the computational times, which suggests that LL1 performs as fast as  the $\ell_1 $ model and at the same time it has the lowest recovery error.

When the noise level is high, for instance $\sigma = 0.1 $, then it is almost impossible to reconstruct the ground-truth signal using any number of measurements. In such cases, our algorithm finds a signal that is sparser and has smaller objective for any choice of the regularization parameter $\gamma $. 


\begin{figure}
\begin{subfigure}{.3\linewidth}
\centering
\includegraphics[width=\textwidth]{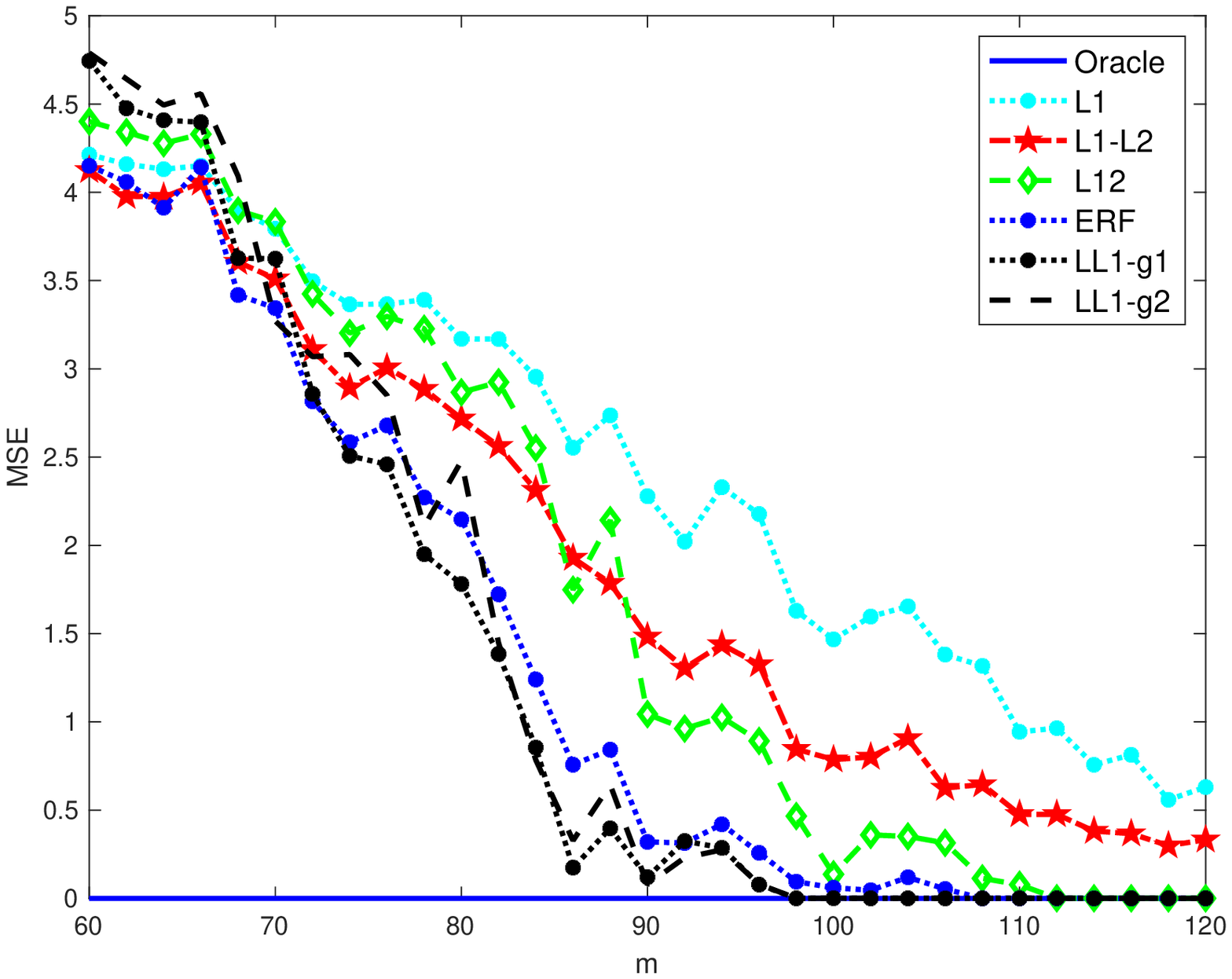}
\caption{$\sigma = 10^{-6}$}
\end{subfigure}%
\hspace{0.5cm}   
\begin{subfigure}{.3\linewidth}
\centering
\includegraphics[width=\textwidth]{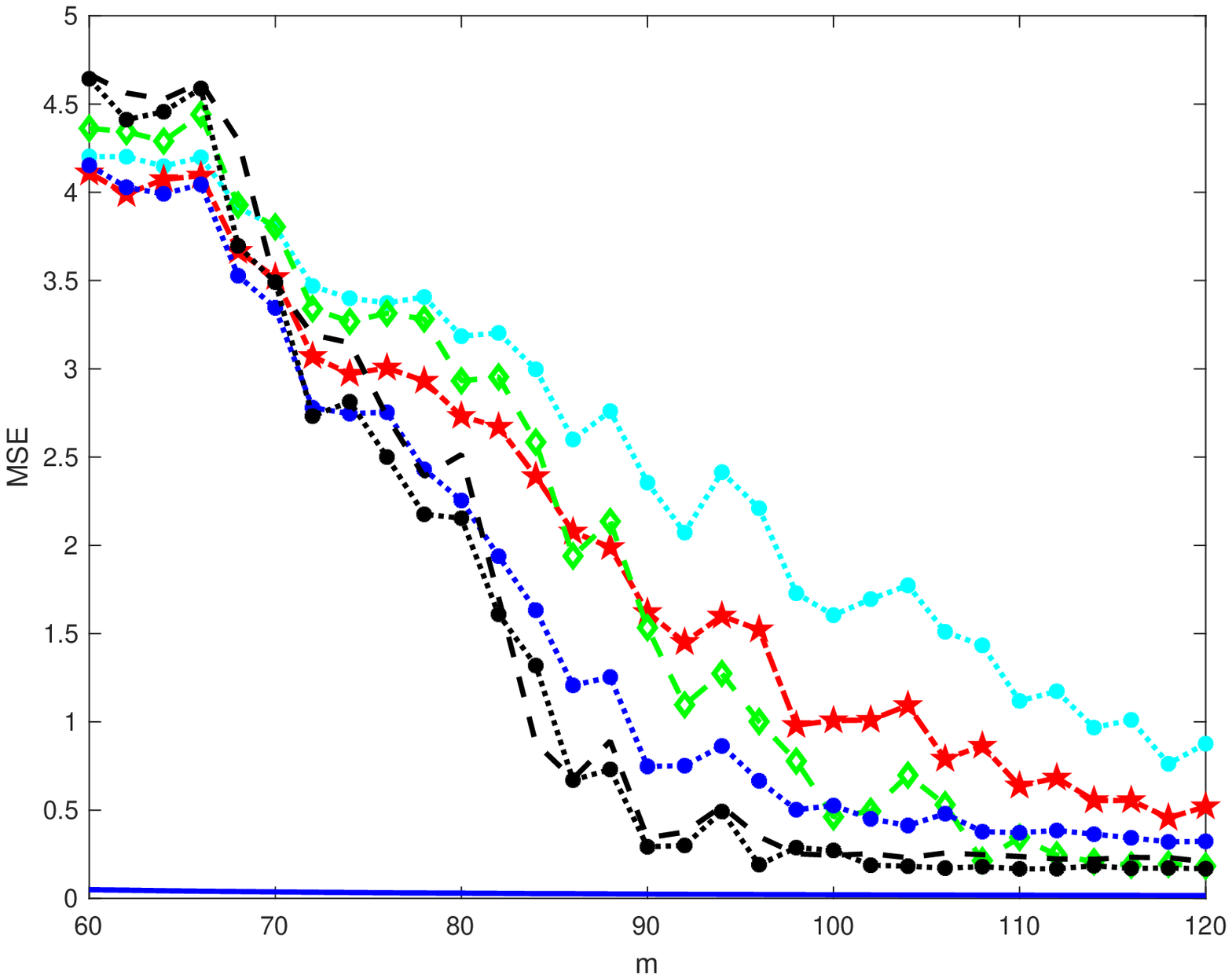}
\caption{$\sigma = 0.01$}
\end{subfigure}
\hspace{0.5cm}   
\centering
\begin{subfigure}{0.3\linewidth}
\centering
\includegraphics[width=\textwidth]{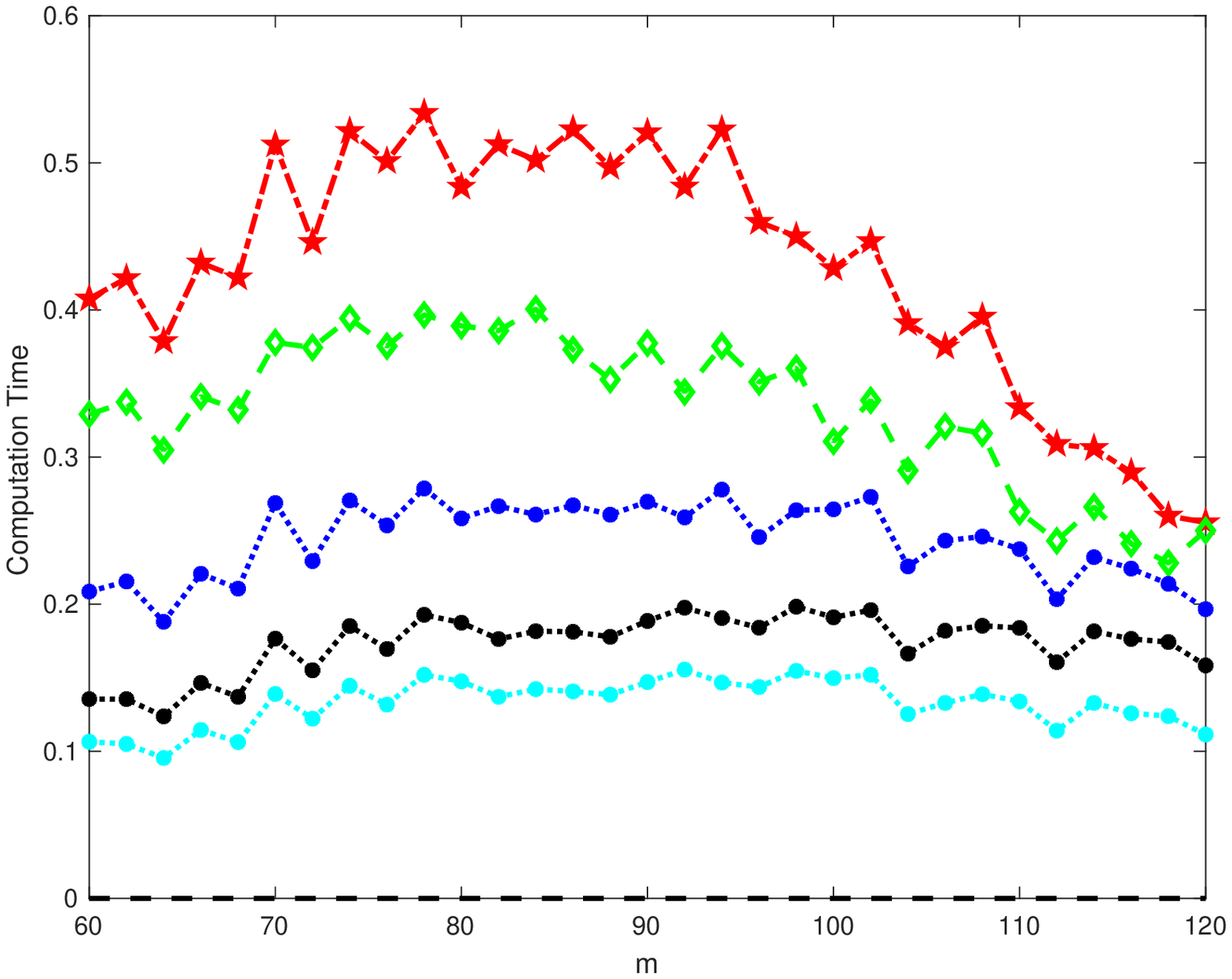}
\caption{Time comparison, $\sigma = 0.01$}
\label{fig:Noisy_comALL01_d}
\end{subfigure}%

\caption{Comparison of all algorithms for $m \times 512 $ matrices with different noise level $\sigma$.}
\label{fig:Noisy_comALL01}
\end{figure}


\subsection{Comparing the ADMM and DCA based algorithms}\label{sect:DCA}

Alternatively to ADMM, one can minimize the proposed model \eqref{eq:con} and \eqref{eq:uncon} by the difference of convex algorithm (DCA) \cite{horst1999dc,taoA97,taoA98}. More specifically, since our function $ -F_{g,\alpha}^U$ is convex on the positive cone, then we use the algorithm introduced in \cite{ochs2015iteratively}, where we only need to find the sub-differential of the function on $\mathbb{R}^n_+ $. We used the function $\psi(\cdot)$ to unify the constrained  and the unconstrained formulations and we have the model
\begin{equation}
\arg \min_{\h x \in \mathbb{R}^n}  F_{g,\alpha}^U(\h x) + \psi(\h x).
\end{equation}
Since the function $F_{g,\alpha}^U $ is concave on $\mathbb{R}^n_+ $, it can be  written as a  difference of two convex functions, i.e.,  $F_{g, \alpha}^U(\h x) +  \psi(\h x) = h_1(\h x) - h_2(\h x) $ where $h_1(\h x) = \psi(\h x) + \frac{\beta}{2} \|\h x \|_2^2 $ and $ h_2(\h x) = \frac{\beta}{2} \| \h x \|_2^2 - F_{g, \alpha}^U (\h x) $ for $\beta \geq 0.$ 
%
An interesting fact about using a DCA form is that if $g $ is a Type B or Type C  then for $\h x \geq 0 $ we have  sub-differentials of the form 
\begin{equation*}
\partial  (-F_{g, \alpha}^U) (\h x)  = - \arg\min_{\h u \in U}   \left\langle \h u, |\h x | \right\rangle + \alpha g(\h u).
\end{equation*}
For $\h x \in \mathbb{R}^n $, take $\h u_{g, \alpha}(\h x^k) \in \arg\min_{\h u \in U}   \left\langle \h u, |\h x^k | \right\rangle + \alpha g(\h u) $ then
the DCA iterations become  
 \begin{equation}
\h x^{k+1} = \arg\min_{\h x \in \mathbb{R}^n} \psi(\h x) + \frac{\beta}{2} \|\h x \|_2^2  - \left\langle \beta |\h x^k| - \h u_{g, \alpha}(\h x^k) ,| \h x|  \right\rangle.
\label{eq:DCA02}
\end{equation}
We implement the DCA iterations of \eqref{eq:DCA02} for $\beta = 0 $ for its simplicity and efficiency as opposed to $ \beta > 0$. In addition,  we can consider 
an adaptive scheme to update $\alpha,$  which is adopted in \cref{Alg:02}. 
 

\begin{algorithm}[t]
	\SetAlgoLined
	Input: $\h x^0 $, $ \alpha^0 > 0$, $ \gamma \geq 0$, function $ g$, set $U $, and MaxOuter/MaxInner   \\
	\For{ $j  = 1$ to $ MaxOuter$}{
		$\alpha^{j+1} = (1 - \eta) \alpha^j $ \\
	\For{$k = 1 $ to $MaxInner$}{
$
 \h u^k =  \arg\min_{\h u \in U}   \left\langle \h u, |\h x^k | \right\rangle + \alpha^j g(\h u).
$
\\
$\h x^{k+1} =\arg\min_{\h x \in \mathbb{R}^n} \psi(\h x) + \left\langle \h u_{g, \alpha}(\h x^k), |\h x | \right\rangle.$
	}}
	\caption{Homotopy based DCA algorithm}
	\label{Alg:02}
\end{algorithm}

We compare 
ADMM (\cref{Alg:01}) and DCA   (\cref{Alg:02}) for minimizing the same constrained formulation \eqref{eq:con}  with  $g_1$ and $g_2$ discussed in \Cref{sec_CompG}. 
We are particularly interested in the algorithmic behaviors when dynamically updating $\alpha$. As mentioned in  \cref{theorem:exact01},  $\alpha$ is supposed to be small enough to approximate the $\ell_0$ solution.
A common way involves an exponential decay in the form of $\alpha^{k+1} = (1 - \eta) \alpha^k,$ for $\eta \in (0,1)$. If the parameter $ \eta$ is close to $ 1$ then $\alpha $ converges to zero too quickly and hence the algorithm cannot converge to a good local minimum as it is equivalent to having $ \alpha = 0$ in the very beginning. On the other hand, if $ \eta$ is close to $0 $ then $\alpha $ slowly decreases to zero; and as a result, \Cref{Alg:01} may terminate before $\alpha$ is small enough for $F^U_{g, \alpha}$ to approximate the $\ell_0$ norm. 


The comparison between ADMM and DCA 
on Gaussian and DCT matrices is presented 
in \cref{fig:SR_All_g1g2}. By fixing $\eta = 0.1,$ we select the optimal  $\rho = 2^j $ for ADMM with $j \in  \{-1,0,1,2,3,4,5,6,7,8 \}$ that achieves the smallest relative error to the ground-truth.  Then using the optimal parameters $\rho$ and $\gamma,$
\cref{fig:SR_All_g1g2} presents the ADMM results for $\eta \in \{ 0.001,0.01,0.1 \}$ and the DCA ones for $\eta \in \{0.01,0.1 \} $.
For  all the cases, ADMM  is superior to  DCA in that it is less sensitive to $\eta$. In addition, DCA consists of two loops and hence it is generally slower than ADMM. Our experiment shows that a suitable choice for our experiments is $\eta = 0.01 $.








\begin{figure}
\centering
\begin{subfigure}{.3\linewidth}
\includegraphics[width=\textwidth]{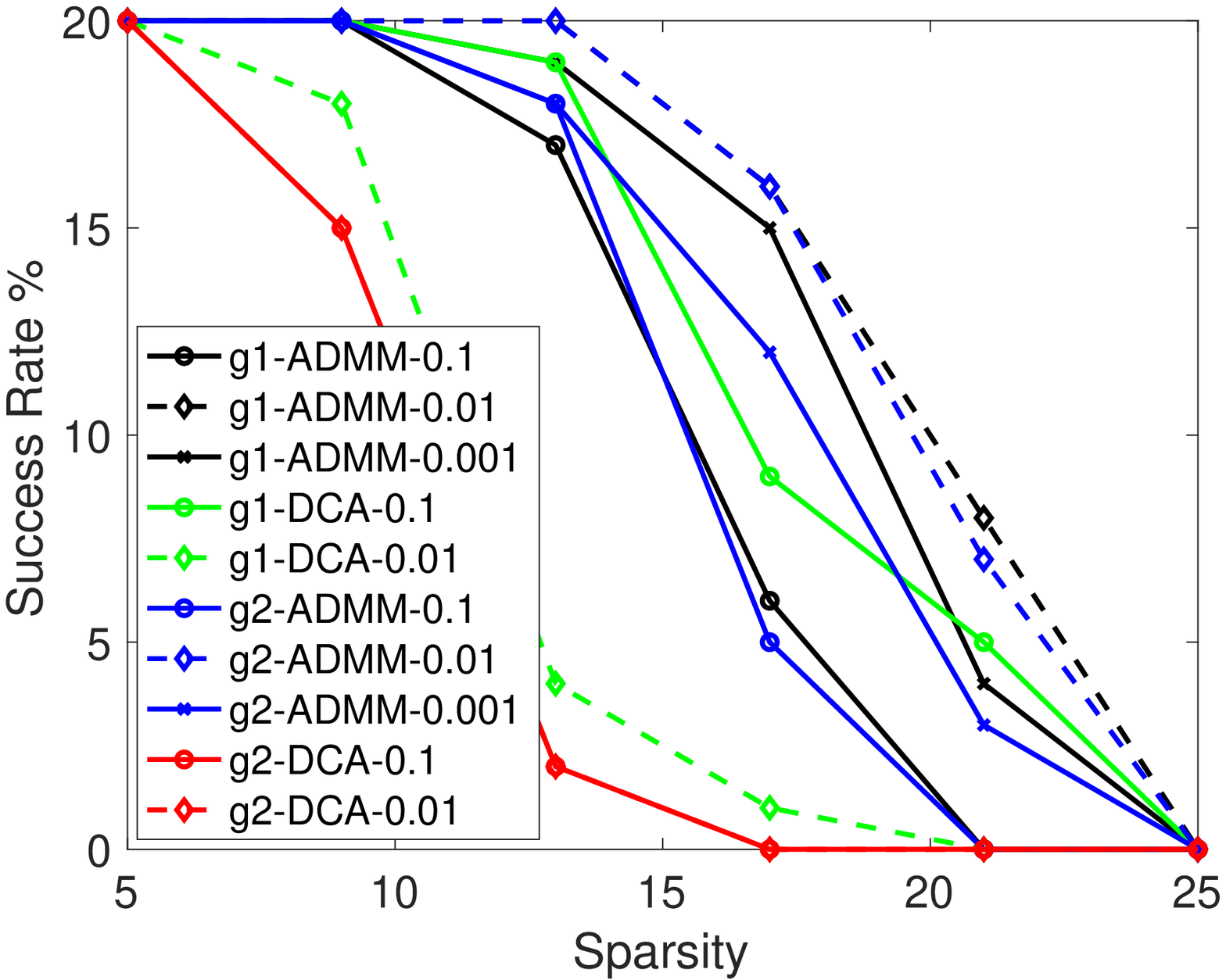}
\caption{$r = 0 $}
\end{subfigure}
~
\begin{subfigure}{.3\linewidth}
\centering
\includegraphics[width=\textwidth]{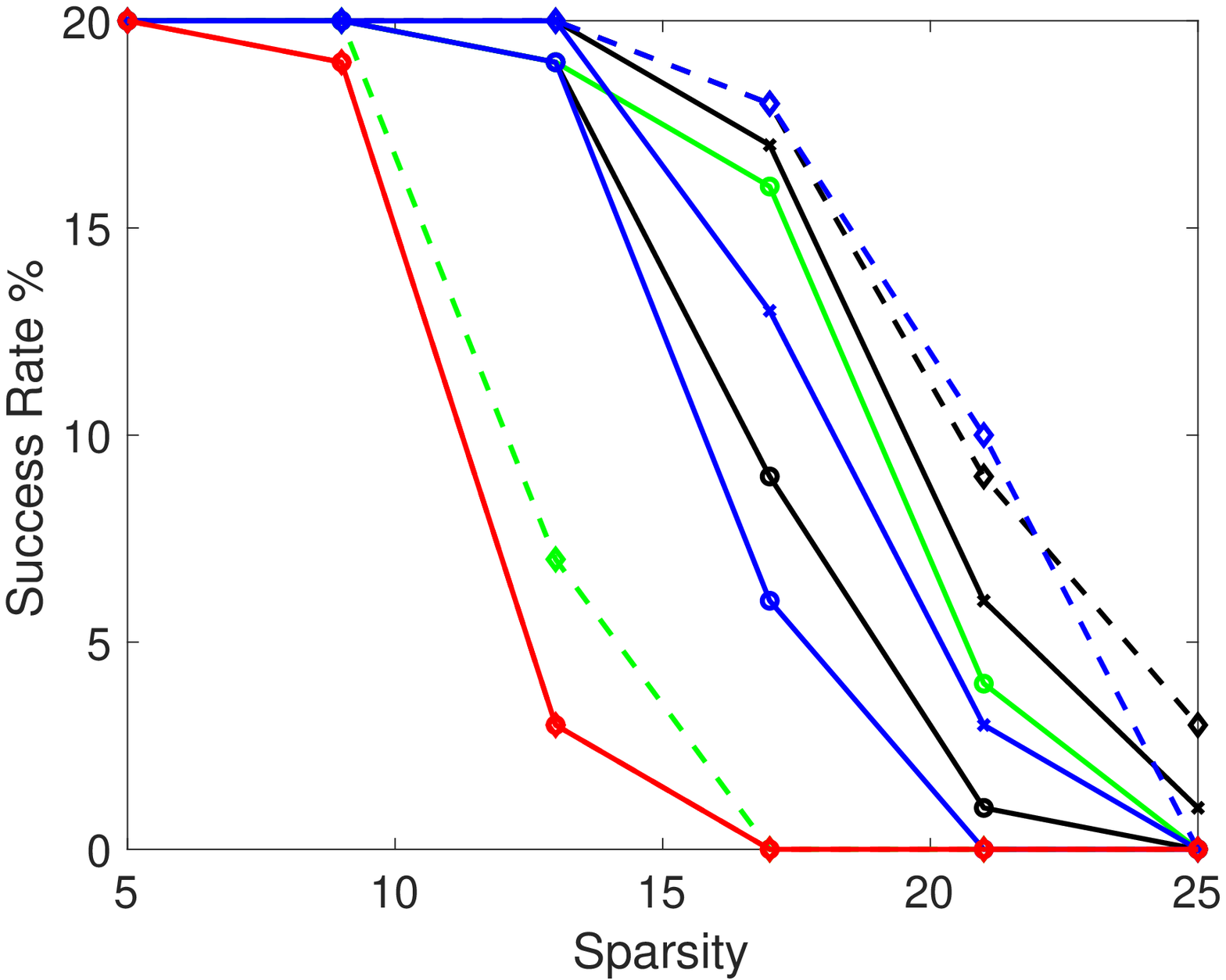}
\caption{$F = 1 $}
\end{subfigure}%
\\
\begin{subfigure}{.3\linewidth}
\centering
\includegraphics[width=\textwidth]{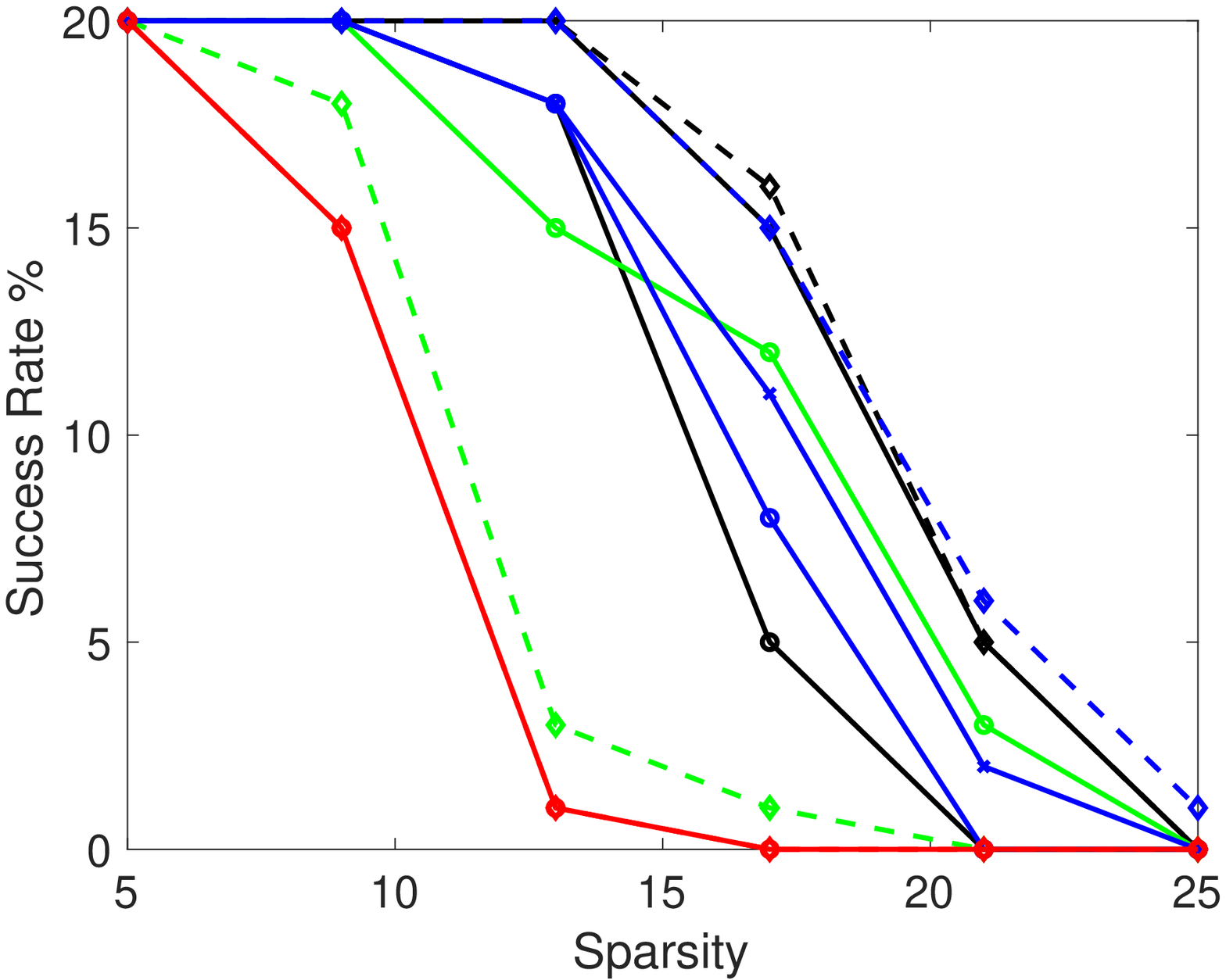}
\caption{$r = 0.8 $}
\end{subfigure}
~
\begin{subfigure}{.3\linewidth}
\centering
\includegraphics[width=\textwidth]{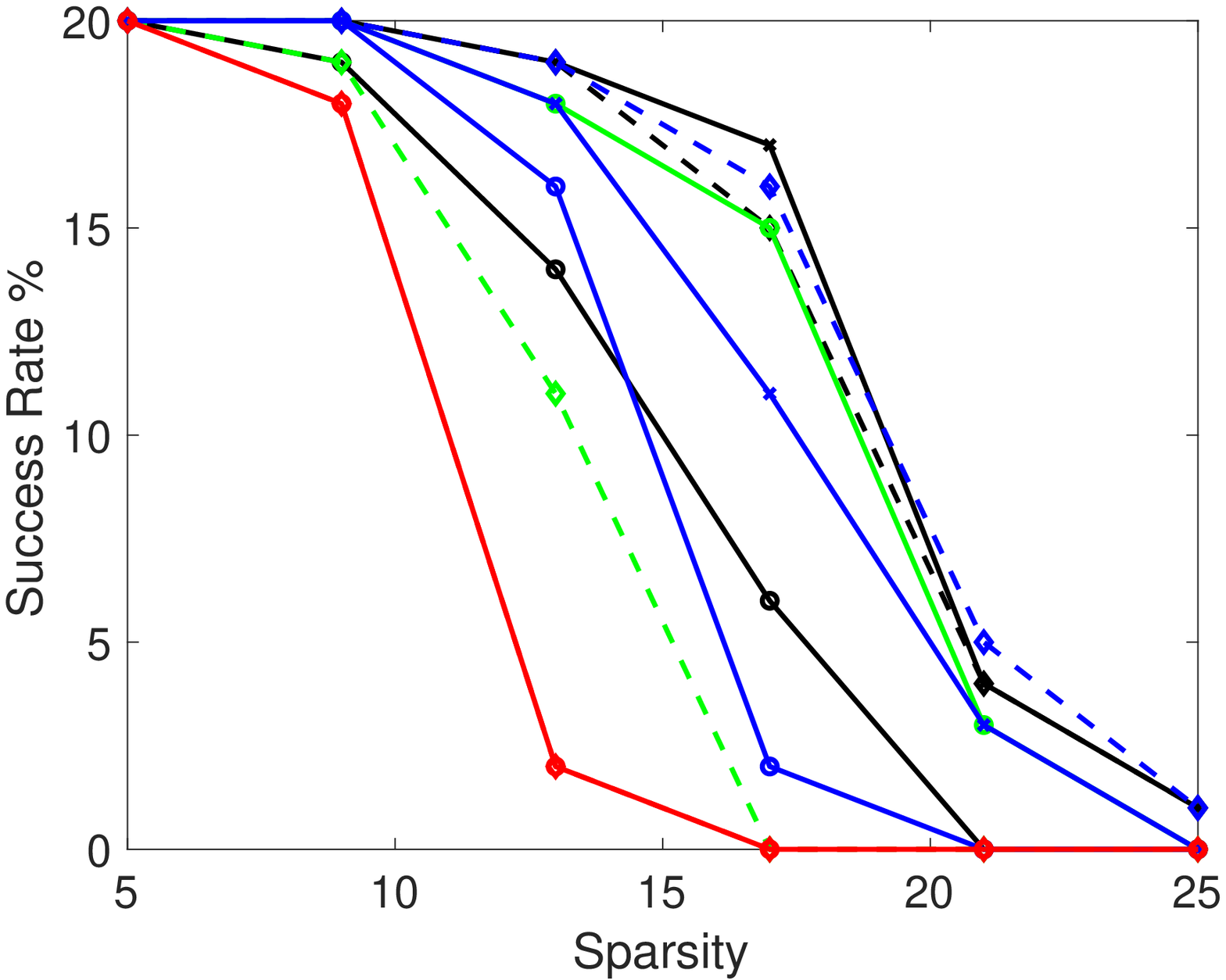}
\caption{$F = 5 $}
\end{subfigure}
\caption{Success rate comparison of ADMM and DCA with two $g$ functions and different $\eta $ for Gaussian matrices (left) with  $r = 0$ and $r = 0.8 $ and DCT matrices (right) with  $F = 1, 5$. }
\label{fig:SR_All_g1g2}
\end{figure}

\section{Concluding Remarks} \label{sec_conc}

In this paper, we proposed a lifted $\ell_1$ model for sparse recovery, which describes a class of  regularizations.  
Specifically we established the connections of this framework to various existing methods that aim to promote sparsity of the model solution.  Furthermore, we proved that our method can exactly recover the sparsest solution under a constrained formulation.  We promoted the use of ADMM to solve for the proposed model with convergence analysis. An alternative approach of using DCA  was discussed in \Cref{sect:DCA}, showing the efficiency of ADMM over DCA. Experimental results on both noise-free and noisy cases illustrate that the proposed framework outperforms the state-of-the-art methods in terms of accuracy and is comparable to the convex $\ell_1$ approach in terms of computational time.


One future work involves the convergence analysis of ADMM for solving the constrained model. One diffciculty lies in the fact that the corresponding function $ \psi(\cdot)$ is a $\delta$-function, which is not differentiable nor \textit{coercive}, and as a result, the proof presented in \Cref{sec_convergence} for the unconstrained minimization is not applicable for the constrained case. 
We observe that the ADMM algorithm for the constrained case does converge and the augmented Lagrangian is decreasing. This empirical evidence suggests the potential to prove the convergence or sufficiently decreasing  of the augmented Lagrangian, which will be left as a future work.








\bibliographystyle{siamplain}
\bibliography{references}

\section*{Supplement}
\label{sec:Supplement}
This section includes the proof of \cref{theorem:Lifting} and the computations to find a lift of sparsity promoting models mentioned in section \ref{sec_connection} into our generalized model.
\begin{proof}[Proof of \cref{theorem:Lifting}]
Set 
\begin{align*}
    f(\h x) = 
    \begin{cases}
    -J(\h x) & \h x \geq \h 0
    \\
    +\infty & \text{otherwise.}
    \end{cases}
\end{align*}
The function $ f$ is proper, lower semi-continuous and convex, hence by \textit{Fenchel-Moreau}'s theorem we have that $f = f^{**} $. Also we have
\begin{align*}
    f^*(\h y) = \sup_{\h x \in \mathbb{R}^n}  \left\langle \h x, \h y \right\rangle - f(\h x) =
    \begin{cases}
     \sup_{\h x \geq \h 0}  \left\langle \h x, \h y \right\rangle + J(\h x) = g(\h y), & \h y \leq \h 0
    \\
    +\infty, & \text{otherwise,}
    \end{cases}
\end{align*}
using $g(\h y) :=  \sup_{\h x \geq \h 0}  \left\langle \h x, -\h y \right\rangle + J(\h x) $.
In addition,
\begin{align*}
    f(\h x) = f^{**}(\h x) = \sup_{\h y \in \mathbb{R}^n}  \left\langle \h x, \h y \right\rangle - f^*(\h y) =
    \begin{cases}
     \sup_{\h y \leq \h 0}  \left\langle \h x, \h y \right\rangle - g(-\h y) & \h x \geq \h 0
    \\
    +\infty & \text{otherwise.}
    \end{cases}
\end{align*}
Therefore,
\begin{align*}
    J(|\h x|) = -f(|\h x|) & = - \sup_{\h y \leq \h 0}  \left\langle |\h x|, \h y \right\rangle - g(-\h y)
     = 
     \inf_{\h y \leq \h 0}  \left\langle |\h x|, -\h y \right\rangle + g(-\h y)
     \\ & = 
     \inf_{\h y \geq \h 0}  \left\langle |\h x|, \h y \right\rangle + g(\h y)
     = 
     \inf_{\h y \in U}  \left\langle |\h x|, \h y \right\rangle + g(\h y),
\end{align*}
where $U \subseteq \{ \h y \geq 0 \mid g(\h y) \neq +\infty \}. $ 
\end{proof}
Given a regularization function $J(\h x),$ we want to find a proper $g$ function and a set $U$ such that $J(\h x)=F_g^U(\h x)$ up to a constant. Using \cref{theorem:Lifting}, we can directly find $F_{g}^U$, however sometimes it might be easier to use the following observation which leads to simpler computations. 
Suppose  $F_{g}^U$  has a unique minimizer $\h u,$ and hence $\h u$ satisfies  $\h u = \nabla_{\h x} F_{g}^U(\h x)= \nabla_{\h x} J(\h x).$  By assuming that the minimum of \eqref{objective:newmodel2} is finite, the optimality condition gives $|\h x| + \nabla_{\h u} g  = 0$ for $\h u \in int(U),$ where $int(U)$ denotes the interior of the set $U $ (Note that  $|\h x| + \nabla_{\h u} g$ can have non-zero coordinates on the boundary of $U$.)
Thus, we only need to solve the following two equations for a function $g $ with respect to $\h u $ on the feasible set   $U$: 
\begin{equation}
\label{eq:Lift}
\begin{cases}
\h u = \nabla_{\h x} J(\h x),
\\
|\h x| + \nabla_{\h u} g = 0, \ \h u \in int(U).
\end{cases}
\end{equation}
\begin{enumerate}
\item \textbf{ $\ell_p $ model:}
Consider $J = J^{\ell_p} / p $, and note that $\frac{\partial J }{\partial |x_i|} = |x_i|^{p-1} $. For $g(\h u) = \sum g_i(u_i) $ and $\h x \in \mathbb{R}^n$, the \eqref{eq:Lift} simplifies into 
\begin{equation*}
\begin{cases}
 u_i = |x_i|^{p-1},
\\
| x_i| +  g_{i}'(u_i) = 0,
\end{cases}
\end{equation*}
for all $i $. From the first equation we get that $|x_i| = u_i^{\frac{1}{p-1}} $ and then from the second equation we get 
$ g_i'(u_i) = -u_i^{\frac{1}{p-1}}$
. A solution for $g $ is $g_i(u_i) = \frac{1-p}{p} u_i^{\frac{p}{p-1}} $ for $u_i \geq 0 $. Finally taking $U = \mathbb{R}^n_+ $ and $g(\h u) = \sum_{i} \frac{1-p}{p} u_i^{\frac{p}{p-1}} $, one can check that $F_{g}^U = J $.

\item \textbf{ log-sum penalty:}
Consider $J =  J^{\log}_a $, and note that $\frac{\partial J }{\partial |x_i|} = \frac{1}{|x_i| + a} $. For $g(\h u) = \sum g_i(u_i) $ and $\h x \in \mathbb{R}^n$, the \eqref{eq:Lift} simplifies into 
\begin{equation*}
\begin{cases}
 u_i = \frac{1}{|x_i| + a},
\\
| x_i| +  g_{i}'(u_i) = 0,
\end{cases}
\end{equation*}
for all $i $. From the first equation we get that $|x_i| = \frac{1}{u_i} - a $ and then from the second equation we get 
$ g_i'(u_i) = a - \frac{1}{u_i}$. A solution for $g $ is $g_i(u_i) = a u_i - \log(u_i) $ for $u_i > 0 $. Finally taking $U = \mathbb{R}^n_{>0} $ and $g(\h u) = \sum_{i} \left( a u_i - \log(u_i) \right) $, one can check that $F_{g}^U + 1 = J $.

\item \textbf{ Smoothly clipped lasso model:}
Consider $J =  J_{\gamma, \lambda}^{\text{SCAD}} $, and note that
\begin{equation}
 \frac{\partial f_{\lambda, \gamma} ^{\text{SCAD}}(t) }{\partial |t|}= 
\left\{
	\begin{array}{ll}
		\lambda & \mbox{if } |t| \leq  \lambda,
		 \\
   \frac{\lambda \gamma - t}{\gamma - 1} & \mbox{if } \lambda < |t| \leq  \gamma \lambda,
		\\
	0 & \mbox{if }  |t| > \gamma \lambda.
	\end{array}
\right.
\end{equation}
For $g(\h u) = \sum g_i(u_i) $ and $\h x \in \mathbb{R}^n$, the first equation in \eqref{eq:Lift} simplifies into 
\begin{equation*}
 u_i = 
 \left\{
	\begin{array}{ll}
		\lambda & \mbox{if } |x_i| \leq  \lambda,
		 \\
   \frac{\lambda \gamma - |x_i|}{\gamma - 1} & \mbox{if } \lambda < |x_i| \leq  \gamma \lambda,
		\\
	0 & \mbox{if }  |x_i| > \gamma \lambda.
	\end{array}
\right.
\end{equation*}
for all $i $. In the case of $\lambda < |x_i| \leq  \gamma \lambda, $ we get that 
$|x_i| = \gamma \lambda - (\gamma - 1) u_i, $
which means we should have $ g_{i}'(u_i) = -\gamma \lambda + (\gamma - 1) u_i$  and $u_i\leq \lambda$. 
By taking $U = [0,\lambda]^n $ and $g(\h u) = \sum_{i} \left( -\gamma \lambda u_i + (\gamma - 1) \frac{u_i^2}{2}  \right) $, one can check that $F_{g}^U + \frac{(\gamma + 1)\lambda^2}{2} = J  $.

\item \textbf{ Mini-max concave penalty:}
Consider $J =  J_{\gamma, \lambda}^{\text{MCP}} $, and note that
\begin{equation}
 \frac{\partial f_{\lambda, \gamma} ^{\text{MCP}}(t) }{\partial |t|}= 
\left\{
	\begin{array}{ll}
		\lambda - \frac{t}{\gamma}  & \mbox{if } |t| \leq \gamma \lambda, \\
		0 & \mbox{if }  |t| > \gamma \lambda.
	\end{array}
\right.
\end{equation}
For $g(\h u) = \sum g_i(u_i) $ and $\h x \in \mathbb{R}^n$, the first equation in \eqref{eq:Lift} simplifies into 
\begin{equation*}
 u_i = 
 \left\{
	\begin{array}{ll}
		\lambda - \frac{|x_i|}{\gamma}  & \mbox{if } |x_i| \leq \gamma \lambda, \\
		0 & \mbox{if }  |x_i| > \gamma \lambda.
	\end{array}
\right.
\end{equation*}
for all $i $.  When $ |x_i| \leq  \gamma \lambda,$ we obtain 
$|x_i| = \gamma (\lambda - u_i),$ which implies that $ g_{i}'(u_i) = -\gamma (\lambda - u_i)$ from the second equation in \eqref{eq:Lift}.
Therefore, we set 
$U = [0,\infty)^n $ and $g(\h u) = \sum_{i} \left( -\gamma (\lambda u_i - \frac{u_i^2}{2})  \right) $, leading to $F_{g}^U + \frac{1}{2} \gamma \lambda^2 = J  $.

\item \textbf{ Capped $\ell_1 $ model:}
Consider $J =  J^{\text{CL1}}_a $, and note that
\begin{equation}
 \frac{\partial J}{\partial |x_i|} = 
\left\{
	\begin{array}{ll}
		1  & \mbox{if } |x_i| < a, \\
		0 & \mbox{if }  |x_i| > a.
	\end{array}
\right.
\end{equation}
 For $g(\h u) = \sum g_i(u_i) $ and $\h x \in \mathbb{R}^n$, the first equation in \eqref{eq:Lift} simplifies into 
\begin{equation*}
 u_i = \left\{
	\begin{array}{ll}
		1  & \mbox{if } |x_i| < a, \\
		0 & \mbox{if }  |x_i| > a.
	\end{array}
\right.
\end{equation*}
for all $i $. Note that the second equation in \eqref{eq:Lift} only happens if the minimizer is in the interior of the set $U $.  Consider $U = [0,1]^n $, therefore for this case since the minimizer is on the boundary, therefore we need a $g $ function which is nonzero in the interior of $ U $ and for $|x_i| < a $ we have  $ | x_i| +  g_{i}'(u_i) < 0$ and for $|x_i| > a $ we have  $ | x_i| +  g_{i}'(u_i) > 0$. Therefore $g_{i}'(u_i) = -a $ and a solution for this is $ g_{i}(u_i) = -a u_i $. Finally taking $U = [0,1]^n $ and $g(\h u) = \sum_{i} \left( -a u_i  \right) $, one can check that $F_{g}^U + a = J  $.

\item \textbf{ Transformed $\ell_1 $ model:}
Consider $J = J_a^{\text{TL1}} / (a+1) $, and note that $\frac{\partial J }{\partial |x_i|} = \frac{a}{(a + |x_i|)^2} $. For $g(\h u) = \sum g_i(u_i) $ and $\h x \in \mathbb{R}^n$, the \eqref{eq:Lift} simplifies into 
\begin{equation*}
\begin{cases}
 u_i =  \frac{a}{(a + |x_i|)^2},
\\
| x_i| +  g_{i}'(u_i) = 0,
\end{cases}
\end{equation*}
for all $i $. From the first equation we get that $|x_i| = \sqrt{\frac{a}{u_i}} - a $ and then from the second equation we get 
$ g_i'(u_i) = a - \sqrt{\frac{a}{u_i}}$. A solution for $g $ is $g_i(u_i) = a u_i - 2 \sqrt{a u_i} $ for $u_i \geq 0 $. Finally taking $U = \mathbb{R}^n_+ $ and $g(\h u) = \sum_{i} a u_i - 2 \sqrt{a u_i} $, one can check that $ F_{g}^U + 1 = J $.

\item \textbf{ Error function penalty:}
Consider $J = J_{\sigma}^{\text{ERF}} $, and note that $\frac{\partial J }{\partial |x_i|} = e^{-x_i^2/\sigma^2} $ and $e^{-x_i^2/\sigma^2} \in (0,1] $. For $g(\h u) = \sum g_i(u_i) $ and $\h x \in \mathbb{R}^n$, the \eqref{eq:Lift} simplifies into 
\begin{equation*}
\begin{cases}
 u_i =  e^{-x_i^2/\sigma^2},
\\
| x_i| +  g_{i}'(u_i) = 0,
\end{cases}
\end{equation*}
for all $i $. From the first equation we get that $|x_i| = \sigma\sqrt{-\log(u_i)} $ and then from the second equation we get 
$ g_i'(u_i) =  -\sigma \sqrt{- \log(u_i)}$. A solution for $g $ is $g_i(u_i) = \sigma \int_{u_i}^{1} \sqrt{- \log(\tau )} d \tau $ for $u_i \in (0,1] $. Finally taking $U = [0,1]^n $ and $g(\h u) = \sigma \sum_{i} \int_{u_i}^{1} \sqrt{- \log(\tau )} d \tau $, one can check that $ F_{g}^U = J $. 

\end{enumerate}


\end{document}